\documentclass[%
reprint,
superscriptaddress,
 amsmath,amssymb,
 aps,
pra,
]{revtex4-2}

\usepackage{graphicx}
\usepackage{dcolumn}
\usepackage{bm}
\usepackage{xcolor}
\usepackage{latexsym}
\usepackage{caption}
\usepackage{subcaption}
\usepackage{parskip}
\usepackage{braket}
\usepackage{physics}
\usepackage{cancel}
\usepackage{balance}
\usepackage{amsthm}
\usepackage{multirow}
\newtheorem{theorem}{Theorem}

\begin{document}
\date{\today}

\title{Decomposition of High-Rank Factorized Unitary Coupled-Cluster Operators \\Using Ancilla and Multi-Qubit Controlled Low-Rank Counterparts}

\author{Luogen Xu}
\affiliation{Department of Physics, Georgetown University\\ 37$^{\rm th}$ and O Sts. NW, Washington, DC 20057 USA}

\author{Joseph T. Lee}
\affiliation{Department of Applied Physics and Mathematics, Columbia University, 500 W. 120th St., New York, NY 10027 USA}

\author{J. K Freericks}
\affiliation{Department of Physics, Georgetown University\\ 37$^{\rm th}$ and O Sts. NW, Washington, DC 20057 USA}

\begin{abstract}
The factorized form of the unitary coupled-cluster (UCC) approximation is one of the most promising methodologies to prepare trial states for strongly correlated systems within the variational quantum eigensolver (VQE) framework.  The factorized form of the UCC ansatz can be systematically applied to a reference state to generate the desired entanglement. The difficulty associated with such an approach is the requirement of simultaneously entangling a growing number of qubits, which quickly exceeds the hardware limitations of today's quantum machines. In particular, while circuits for singles and double excitations can be performed on current hardware, higher-rank excitations require too many gate operations. In this work, we propose a set of new schemes that trade off using extra qubits for a reduced gate depth to decompose high-rank UCC excitation operators into significantly lower depth circuits. These results will remain useful even when fault-tolerant machines are available to reduce the overall state-preparation circuit depth.
\end{abstract}

\maketitle

\section{Introduction}
Efficiently simulating quantum many-body systems on quantum hardware is one of the major goals of quantum computation and many algorithms already exist \cite{preskill_2018, aspuru-guzik_2005, lloyd_1996, lee_huggins_head-gordon_whaley_2018} for this problem. For weakly correlated systems seen in many quantum chemistry systems, there is a hierarchy to the amplitudes of the determinants in the expansion of the ground-state wavefunction---low-rank excitations from the reference state typically have larger amplitudes than higher-rank excitations. But, generically, many determinants are still needed to achieve chemical accuracy even with this hierarchy. When the number of electrons and spin orbitals is small enough, the molecule can be treated by exact diagonalization, which is called full configuration interaction (FCI) \cite{fci} in the chemistry field. But very few systems can be treated this way on classical computers due to the exponential growth of the Hilbert space scaling like $\mathcal{O}(2^N)$. Truncating the Hilbert space to include the most important many-body basis states is called the configuration interaction (CI) method. But it suffers from not being size consistent, which affects the accuracy when molecules are stretched close to the dissociation limit. Instead, the coupled cluster (CC) method \cite{bartlett_2007} provides high precision, is size consistent, and is lean on memory usage, because it does not explicitly construct the wavefunction. The CC method scales like  $\mathcal{O}(N^{10})$, when including singles, doubles, triples, and quadruples. 

The variational quantum eigensolver (VQE) algorithm relies on the variational principle of quantum mechanics to estimate the ground-state energy of a molecule \cite{vqe}. While the VQE can be used for physical systems in condensed matter and other fields of physics, the main application is in quantum chemistry. Classical quantum chemistry methods boast high accuracy, but can be expensive. Thus, one active area of research is in leveraging quantum technology to calculate the ground-state energy of molecules.

Picking the proper wavefunction ansatz is one of the more difficult challenges in not only using VQE for quantum chemistry, but also in other approaches where a trial wavefunction is needed \cite{stair_evangelista_2021}. The classical coupled-cluster approximation applies an exponential operator to the reference state (typically the Hartree-Fock wavefunction). In conventional CC calculations, one applies the coupled-cluster operator as a similarity transformation of the Hamiltonian. Because the Hamiltonian only contains single and two-particle interactions, the power series expansion of the similarity transformation truncates after at most four-fold nested commutator terms, which proves efficient when carried out on classical computers. However, most operations applicable to quantum machines must be unitary. This suggests using the unitary coupled cluster (UCC) ansatz \cite{bartlett_kucharski_noga_1989, schaefer_2013}, which includes a sum of excitations minus de-excitations, to have a unitary operator applied to the reference state.  Unfortunately, the similarity transformation of the Hamiltonian under the UCC ansatz does not truncate after a small number of terms. Strategies used to evaluate it on classical computers include truncating the series at a fixed order \cite{bartlett_kucharski_noga_1989}, expanding the exponential operator in a power series and truncating the series when higher-order terms no longer change the wavefunction \cite{evangelista_2011}, and using an exact operator identity of the factorized form of the UCC to allow the wavefunction to be constructed in a tree structure \cite{chen_cheng_freericks_2021}. If $\hat{T}$ is the operator that is exponentiated in the traditional CC approach and applied to the reference state, \textit{i.~e.,} $\ket{\psi_{CC}} = \exp( \hat{T})\ket{\psi_0}$, the corresponding UCC ansatz is the unitary variant, given by $\ket{\psi_{UCC}} = \exp(\hat{T} - \hat{T}^{\dagger})\ket{\psi_0}$. Note that the $\hat{T}$ operator involves fermionic destruction operators for the real orbitals (present in the reference state) and an equal number of fermionic creation operators for the virtual operators (corresponding to orbitals that can be occupied in the expansion of the wavefunction). So, the $\hat{T}$ operator excites the reference state. The operator $\hat{T}^\dagger$ annihilates against the reference state, but it can be nonzero when it acts on other determinants in the expansion for the wavefunction. The standard way to implement the UCC approximation is to  exponentiate the sum of all the different excitation and de-excitation operators $\hat{T} - \hat{T}^\dagger$ via
\begin{equation}
    \hat U_{UCC} =e^{\hat{T}-\hat{T}^\dagger}= e^{\sum \theta_{ijk\cdots}^{abc\cdots}\big[\hat A_{ijk\cdots}^{abc\cdots} - \big({\hat A_{ijk\cdots}^{abc\cdots}\big)^\dagger}\big]},
    \label{eq:fulucc}
\end{equation}
where we define the excitation operators as 
\begin{equation}
    \hat A_{ijk\cdots}^{abc\cdots} = \hat a_a^\dagger \hat a_b^\dagger \hat a_c^\dagger \cdots \cdots \hat a_k \hat a_j \hat a_i.
    \label{eq:op_def}
\end{equation}
Here, $a,b,c,\cdots$ are the indices for the unoccupied (virtual) spin-orbitals, and $i,j,k,\cdots$ are the indices for the occupied (real) spin orbitals and we use the standard second quantization notation for the fermionic  creation and destruction operators; note that in each $\hat{T}$ operator, all creation operators and all destruction operators are selected from the virtual orbitals or the real orbitals, respectively, and \textit{vice versa} for the $\hat{T}^\dagger$ operators. 

It is important to note that carrying out a UCC calculation exactly using this method is challenging as quantum circuits for the exponential of the sum of unitaries are complicated \cite{childs_kothari_somma_2017}. The other method is to write down the ansatz in a factorized form, given by
\begin{equation}
    \hat U^\prime _{UCC} = \prod e^{\theta_{ijk\cdots}^{abc\cdots}\big[\hat A_{ijk\cdots}^{abc\cdots} - \big({\hat A_{ijk\cdots}^{abc\cdots}\big)^\dagger}\big]}.
    \label{eq:factorizeducc}
\end{equation}
Unlike the case where we exponentiate a sum of unitaries, the factorized ansatz is not uniquely determined because many of the elementary factors do not commute leading to different results based on the ordering of the different factors. Despite this, the factorized form is still a promising approach for applying the UCC ansatz on NISQ machines. One reason is it can be implemented with relatively simple circuits. Using the Jordan-Wigner transformation \cite{jordan_wigner_1928, nielsen_2005}, we convert the $\hat{A}$ and $\hat{A}^\dagger$ terms  into sums of products of Pauli strings; one can immediately map the exponential of such operators into a gated circuit. The circuit requires many CNOT gate cascades \cite{barkoutsos_2018, romero_babbush_mcclean_hempel_love_aspuru-guzik_2018}, which will lead to low fidelity performance on current quantum hardware. Reducing the CNOT count of the factorized form of the UCC ansatz could potentially allow for the use of NISQ hardware for quantum chemistry calculations.

Traditional quantum chemistry focuses primarily on singles and doubles excitations in CC, but as the correlations grow, it is anticipated that higher-rank excitations will be needed to accurately represent the wavefunction. Within the classical computational chemistry framework, work by Chen et al. \cite{chen_cheng_freericks_2021} created an algorithm using the factorized form of the UCC that produces significantly better results for strongly correlated systems and comparable results in terms of accuracy for weakly correlated systems. In reference \cite{evangelista_chan_scuseria_2019}, Evangelista, et al. have proved that the disentangled (factorized) UCC ansatz is capable of generating arbitrary states. Ref.~\onlinecite{xu_lee_freericks_2020} shows one can create the exact ground state wavefunction for a four-site Hubbard ring (in its natural orbital basis) using a factorized form of the UCC that requires one quadruple excitation and eight double excitations. Although the circuit depth for such a state preparation procedure is comparatively low, the one quadruple factor requires about half the gate counts for the circuit (being about an order of magnitude more gates than one doubles factor). In this work, we introduce a decomposition method that greatly reduces the gate count of costly high-rank UCC factors (such as the quadruple excitation aforementioned) into lower-rank factors. 

It is important to mention that the method proposed by this paper is predicated on the fact that the fermion-to-qubit mapping used by the circuit from reference \cite{evangelista_chan_scuseria_2019} is the Jordan-Wigner encoding. It is not universally applicable to other encodings. However, one should be able to generalize the approach given here to other fermion encodings, if desired.

\section{Background}
\subsection{Classical coupled-cluster approach}
A set of electronic excitation operators can be defined as follows \cite{helgaker_2014}:
\begin{equation}
    \hat T = \sum_{i=1}^N \hat T_i 
\end{equation}
Explicitly, the first two ranks (orders) are
\begin{align}
    \hat T_1 = \sum_{ia} \theta_i ^a \hat a_a ^\dagger \hat a_i  = \sum_{ia} \theta_i^a \hat A_i^a\\
    \hat T_2 = \sum_{ijab} \theta_{ij}^{ab} \hat a_a ^\dagger \hat a_b ^\dagger \hat a_i \hat a_j = \sum_{ijab} \theta_{ij}^{ab} \hat A_{ij}^{ab}
\end{align}
where $\hat a_a^{\dagger}$ is the fermionic creation operator on virtual orbital $a$ and $\hat a_i$ is the fermionic anihilation operator on real orbital $i$, and they obey the anti-commutation relations as follows:
\begin{equation}\label{ccr}
    \{\hat a_i, \hat a_j\} = 0; \{\hat a_i^{\dagger}, \hat a_j^{\dagger}\} = 0; \{\hat a_i, \hat a_j^\dagger\} = \delta_{ij} 
\end{equation}
where $\{A, B\}=AB+BA$ and $\delta_{ij}$ is the Kronecker delta function.Note that for $\hat T_2$ and higher-rank operators, different ordering of the indices ${ijab}$ can be used but in this work, we will only be using one ordering of the indices for each equivalent term. A coupled-cluster singles and doubles (CCSD) wavefunction is given by an exponential of the excitations acting on a reference state (Hartree-Fock solution)
\begin{equation}
    \ket{\psi_{CCSD}} = e^{\hat T_{CCSD}}\ket{\psi_0} = e^{\hat T_1 +\hat T_2}\ket{\psi_0}
\end{equation}
We compute the energy by first projecting the Schrodinger equation $H\ket{\psi_{CCSD}}=E\ket{\psi_{CCSD}}$ onto the HF reference $\bra{\psi_0}$:
\begin{equation}
    E = \bra{\psi_0}e^{-\hat T_{CCSD}} H e^{\hat T_{CCSD}}\ket{\psi_0}
\end{equation}
We then project against a set of states $\{\bra{\psi_\mu}\}$ that covers the entire space generated by $\hat T_{CCSD}$ acting on the reference state \cite{bartlett_2007, helgaker_2014}. The problem is solved by solving a set of non-linear amplitude equations:
\begin{align}
    E = \bra{\psi_0}e^{-\hat T_{CCSD}} H e^{\hat T_{CCSD}}\ket{\psi_0} \label{ccsdeq1} \\
    0 = \bra{\psi_\mu}e^{-\hat T_{CCSD}} H e^{\hat T_{CCSD}}\ket{\psi_0}. \label{ccsdeq2}
\end{align}
The cost of solving these equations scales as $\mathcal{O}(\eta^2(N-\eta)^4)$, where $\eta$ is the number of electrons and $N$ is the number of spin orbitals in the system. Note the number of amplitude equations is given by the number of amplitudes in the expansion of the $\hat{T}$ operator, which is a much smaller number than the total number of determinants in the CC wavefunction.

It is convenient that the operator $e^{-\hat T_{CCSD}}H e^{\hat T_{CCSD}}$, also known as the similarity transformed Hamiltonian, is additively separable. Combined with the fact that the exponential of the excitation $e^{\hat T_{CCSD}}$ is multiplicatively separable, the CCSD ansatz is size-consistent. As mentioned previously, classical coupled-cluster theory solves the lack of size-consistency of the truncated CI wavefunctions. Recall that the Hamiltonian in second quantization is
\begin{equation}
    H = \sum_{ij}h_{ij}\hat a_i ^\dagger \hat a_j + \frac{1}{2}\sum_{ijkl}g_{ijkl}\hat a_i ^\dagger \hat a_j ^\dagger \hat a_k \hat a_l,
\end{equation}
where $h_{ij}$ are the one-electron integrals, and $g_{ijkl}$ are the two-electron integrals, given by
\begin{align}
    h_{ij} = \int dr_1 \phi_i^* (r_1) \Bigg(-\frac{1}{2}\nabla_{r_1} ^2 - \sum_{I=1} ^M \frac{Z_I}{R_{1I}} \Bigg)\phi_j (r_1) \\
    g_{ijkl} = \int dr_1 dr_2 \phi_i^*(r_1)\phi_j^*(r_2)\frac{1}{r_{12}}\phi_k(r_1)\phi_l(r_2).
\end{align}
Here, $M$ is the number of atoms in the system, $Z_I$ are atomic numbers, $R_{1I} = \abs{r_1 - R_I}$, $r_{12} = \abs{r_1 - r_2}$, and $\phi(r)$ are mean-field solutions such as HF \cite{szabo_ostlund_2006, taketa_huzinaga_o-ohata_1966}. 
A general similarity transformed Hamiltonian can be expanded using the Hadamard lemma, and it truncates after the fourth term $\frac{1}{24}[[[[H,T],T],T],T]$ due to the Hamiltonian having only one- and two-body interaction terms \cite{bartlett_2007}. However, when acting on a multi-reference state, which is often needed for strongly correlated systems, the calculational procedure often becomes problematic.

\subsection{Unitary coupled-cluster and disentangled ucc factors}
The unitary variant of the CC method is defined as follows \cite{schaefer_2013, bartlett_kucharski_noga_1989}:
\begin{equation}
    \ket{\psi_{UCC}} = e^{\hat T-\hat T^\dagger}\ket{\psi_0}
\end{equation}
The UCC method computes the energy using the variational principle:
\begin{equation}
    E = \min_{\Vec{\theta}}\frac{\bra{\psi_0}e^{-(\hat T-\hat T^\dagger)}He^{\hat T-\hat T^\dagger}\ket{\psi_0}}{\braket{\psi_{UCC}}{\psi_{UCC}}}
\end{equation}
which requires us to work with the explicit wavefunction or to determine the similarity transformation of the Hamiltonian. This approach is always variational, is size-consistent, and often can be extended to multireference situations. However, the Hadamard lemma expansion of its similarity transformed Hamiltonian no longer truncates after just four terms \cite{taube_bertlett_2006, kutzelnigg_1991}. Although the UCC ansatze are challenging to carry out on a classical computer, a quantum computer can efficiently apply a UCC operator in its factorized form \cite{vqe, sokolov_2020}. 

To implement the UCC ansatz on a quantum machine requires Trotterization as the excitation operators do not necessarily commute:
\begin{equation}
    e^{\hat T-\hat T^\dagger} = e^{\sum_i \theta_i (\hat A_i-\hat A_i^\dagger)} \approx \Bigg(\prod_i e^{\frac{\theta_i}{n} (\hat A_i- \hat A_i^\dagger)} \Bigg)^n
    \label{eq:trotter}
\end{equation}
where $\theta_i$ is the amplitude associated with the excitation operator $\hat A_i$ and $\hat A_i^\dagger$. In the case where $n=1$, we can write the UCC ansatz as:
\begin{equation}
    \ket{\psi_{UCC}'} = \prod_i e^{\theta_i (\hat A_i-\hat A_i^\dagger)} \ket{\psi_0}
\end{equation}
where a UCC factor is then of the form $e^{\theta_i (\hat A_i-\hat A_i^\dagger)}$. One can think of this either as a crude approximation to the Trotter product or as a new factorized form of the UCC ansatz. It is important to note that this ansatz is not unique---different orderings leads to different wavefunctions when the re-ordered factors do not commute with each other.

\subsection{SU(2) identity for single UCC factors}
A single UCC factor has a hidden SU(2) identity that exactly determines the exponential of the operator~\cite{evangelista_chan_scuseria_2019,xu_lee_freericks_2020,chen_cheng_freericks_2021}. The identity follows by simply calculating powers of the exponent. We first note that
\begin{gather}
    (\hat A + \hat A^\dagger)^2 = \hat{A}\hat A^\dagger + \hat A^\dagger \hat{A} \nonumber \\
    = \hat n_{a_1}\hat n_{a_2}\cdots\hat n_{a_n}(1-\hat n_{i_1})(1-n_{i_2})\cdots(1-\hat n_{i_n}) \nonumber\\ 
    +(1-\hat n_{a_1})(1-\hat n_{a_2})\cdots(1-\hat n_{a_n})\hat n_{i_1}\hat n_{i_2}\cdots\hat n_{i_n}, 
    \label{eq:square_identity}
\end{gather}
because $\{i,j,k,\cdots\}$ and $\{a,b,c,\cdots\}$ are disjoint sets. Here, $\hat n_\alpha = \hat a_\alpha ^\dagger \hat a_\alpha$ is the number operator for spin-orbital $\alpha$.  The cubed term can then be simplified to be
\begin{align}
    (\hat A+\hat A^\dagger)^3 =\hat A\hat A^\dagger\hat A+\hat A^\dagger\hat A\hat A^\dagger=\hat A+\hat A^\dagger,
    \label{eq:cube_identity}
\end{align}
This makes the power-series expansion of the exponential simple: terms with odd powers are proportional to $\hat A + \hat A^\dagger$ and terms with even powers are proportional to Eq. (\ref{eq:square_identity}). We just have to be careful with the zeroth-power term, which is different. Hence, we have
\begin{gather}
    e^{\theta\left [\hat A_{i_1\cdots i_n}^{a_1\cdots a_n}-
    \left (\hat A_{i_1\cdots i_n}^{a_1\cdots a_n}\right )^\dagger\right ]}\nonumber \\
    =\hat I+\sin\theta \left [\hat A_{i_1\cdots i_n}^{a_1\cdots a_n}-
    \left (\hat A_{i_1\cdots i_n}^{a_1\cdots a_n}\right )^\dagger\right ]\nonumber\\
    +(\cos\theta-1)\left [\hat n_{a_1}\hat n_{a_2}\cdots\hat n_{a_n}(1-\hat n_{i_1})(1-n_{i_2})\cdots(1-\hat n_{i_n})\right .\nonumber\\
    \left .+(1-\hat n_{a_1})(1-\hat n_{a_2})\cdots(1-\hat n_{a_n})\hat n_{i_1}\hat n_{i_2}\cdots\hat n_{i_n}\right ].
    \label{eq:identity}
\end{gather}
This identity implies that when a single UCC factor acts on a state that neither $\hat A$ excites nor $\hat A^\dagger$ de-excites, the state is unchanged by the operator. But when the single UCC factor acts on a state that can be excited by $\hat A$ or de-excited by $\hat A^\dagger$, the result is a cosine multiplied by the original state plus a sine multiplied by the excited (or de-excited) state. It is important to note that the identity, Eq. (\ref{eq:identity}) holds for \textit{any} rank of the UCC factor.

\subsection{Exactness of the factorized UCC circuits}

In this section, we will show that the circuit for a UCC doubles factor is exact. The UCC doubles in the factorized form serve as the cornerstone of this study as we aim to decompose the high-rank operators into ones that contain primarily doubles terms.

The factorized form of the double excitation is written as:
\begin{equation}
    \hat U(\theta) = \exp\Big(\frac{\theta_{ijkl}}{2}\big(\hat a_i ^\dagger \hat a_j ^\dagger \hat a_k \hat a_l - \hat a_l ^\dagger \hat a_k ^\dagger \hat a_j \hat a_i\big) \Big)
    \label{eq:double_factorized}
\end{equation}
Here we define the factorized UCC double excitation using the half angle $\theta_{ijkl}/2$ because this facilitates the correct rotation operators $U_\theta$ used in the quantum circuits in later sections. As discussed before, the product of these factors forms a subspace of the full Hilbert space. Although non-unique, if multiplied in a specific order, the product of these factors can be used to create very accurate trial wavefunctions \cite{evangelista_chan_scuseria_2019, xu_lee_freericks_2020}. To implement the UCC factors presented by Eq. (\ref{eq:double_factorized}) on quantum hardware while fully capturing the anti-commutation relations shown in Eq. (\ref{ccr}), we choose to apply the Jordan-Wigner (JW) transformation to write the fermionic operators in terms of Pauli strings \cite{jordan_wigner_1928, nielsen_2005, somma_ortiz_gubernatis_knill_laflamme_2002}: 

\begin{gather}
    \exp\Big(\frac{\theta_{ijkl}}{2}\big(\hat a_i ^\dagger \hat a_j ^\dagger \hat a_k \hat a_l - \hat a_l ^\dagger \hat a_k ^\dagger \hat a_j \hat a_i \big) \Big) = \exp \Bigg( \frac{i\theta_{ijkl}}{16} \bigotimes _{a=l+1} ^{k-1} Z_{a} \nonumber \\
    \bigotimes _{b=j+1} ^{i-1} Z_{b} \times \bigg( X_l  X_k  Y_j  X_i +  Y_l  X_k  Y_j  Y_i \nonumber \\
    + X_l  Y_k  Y_j  Y_i +  X_l  X_k  X_j  Y_i \nonumber \\
    - Y_l  X_k  X_j  X_i -  X_l  Y_k  X_j  X_i \nonumber \\
    - Y_l  Y_k  Y_j  X_i -  Y_l  Y_k  X_j  Y_i\bigg)\Bigg)
    \label{doublefactor_jw}
\end{gather}
Eq. (\ref{doublefactor_jw}) is obtained by applying the JW transformation to Eq. (\ref{eq:double_factorized}) with the convention $\hat a_n =\frac{1}{2}\big( X + i Y \big)\bigotimes Z ^{\bigotimes N-n-1}$ and $\hat a_n ^\dagger = \frac{1}{2}\big( X -i Y \big)\bigotimes Z^{\bigotimes N-n-1}$, where $ X$, $ Y$, and $ Z$ are the Pauli matrices, and $0 \leq n \leq N-1$, $N$ being the number of qubits. The qubit state $|0\rangle$ has no electrons and $|1\rangle$ has one electron.

\begin{table}[h]
\resizebox{\columnwidth}{!}{
\begin{tabular}{|l|l|l|l|l|l|l|l|l|}
\hline
$\{l,k,j,i\}$ & $XXYX$ & $YXYY$ & $XYYY$ & $XXXY$ & $YXXX$ & $XYXX$ & $YYYX$ & $YYXY$ \\ \hline
$XXYX$        & 0    & 2    & 2    & 2    & 2    & 2    & 2    & 4    \\ \hline
$YXYY$        & 2    & 0    & 2    & 2    & 2    & 4    & 2    & 2    \\ \hline
$XYYY$        & 2    & 2    & 0    & 2    & 4    & 2    & 2    & 2    \\ \hline
$XXXY$        & 2    & 2    & 2    & 0    & 2    & 2    & 4    & 2    \\ \hline
$YXXX$        & 2    & 2    & 4    & 2    & 0    & 2    & 2    & 2    \\ \hline
$XYXX$        & 2    & 4    & 2    & 2    & 2    & 0    & 2    & 2    \\ \hline
$YYYX$        & 2    & 2    & 2    & 4    & 2    & 2    & 0    & 2    \\ \hline
$YYXY$        & 4    & 2    & 2    & 2    & 2    & 2    & 2    & 0    \\ \hline
\end{tabular}
}
\caption{Commutation table for all eight 4-qubit Pauli strings from Eq. (\ref{doublefactor_jw}). Integers in each entry count the number of indices that anticommute.}
\label{table:commutation}
\end{table}
In table \ref{table:commutation}, we show that the number of anticommuting indices between the Pauli strings in Eq.~(\ref{doublefactor_jw}) is always even, which implies every Pauli string commutes with every other Pauli string. This means that the exponential of the sum of the eight Pauli strings can be rewritten as eight products of the exponential of each Pauli string. The ordering of the exponential factors is unimportant, because they all commute with each other. Below we provide a proof of this conclusion.

\begin{theorem}
Consider two Pauli strings acting on the same set of qubits,\\
    \begin{equation}
        P_A = \bigotimes_{i=1}^{N}A_i, P_B = \bigotimes_{i=1}^{N}B_i \nonumber
    \end{equation}
where $A_i,B_i \in \{X,Y, Z,I\}$. $P_A$ and $P_B$ commute iff $A_i$ and $B_i$ anticommute on an even number of indices.
\end{theorem}

\begin{proof}
Pauli matrices that do not commute, anticommute. Therefore, we can write explicitly
\begin{equation}
P_A P_B = \bigotimes_{i=1}^N A_i B_i = \bigotimes_{i=1}^N \nonumber
\begin{cases}
    B_i A_i, & \text{if }  [A_i, B_i]=0  \\
    -B_i A_i, &\text{if }  [A_i, B_i] \neq 0.
\end{cases}
\end{equation}
The two factors $A_i$ and $B_i$ commute, if they are both the same Pauli operator, or if one of them is the identity; otherwise, they anticommute.
In order for $P_A P_B$ to equal $P_B P_A$, there must be an even number of cases where $[A_i, B_i ]\neq 0$ because $(-1)^{2n} = 1$. Therefore $P_A$ and $P_B$ commute iff $A_i$ and $B_i$ anticommute on an even number of indices.
\end{proof}

Since the subterms of the UCC doubles operator all commute, the circuit shown in Fig.~\ref{ucc_doubles}, which implements Eq.~(\ref{doublefactor_jw}), is exact. In fact, a general UCC factor of order $n$ will have $2^{2n-1}$ terms after the JW transformation multiplying strings of Pauli $Z$ operators, among which numbers of Pauli $X$ and $Y$ operators are always odd, making numbers of anticommuting indices always even and thus all the strings that contain $X$ and $Y$ commute with one another \cite{romero_babbush_mcclean_hempel_love_aspuru-guzik_2018}.

\subsection{The Conventional Quantum Circuits}

This section will show how one can construct the circuits for each UCC factor. The standard circuit for a single UCC doubles factor was derived in  Ref.~\onlinecite{barkoutsos_2018, romero_babbush_mcclean_hempel_love_aspuru-guzik_2018} and is shown in Fig.~\ref{ucc_doubles}. As shown in section D, the Pauli strings in the exponentials commute. Therefore, a UCC factor can be rewritten as a product of exponentials of Pauli strings. The circuit for the UCC factors follows a similar prescription to \cite{nielsen_chuang_2019}. Nielsen and Chuang provide a strategy for creating circuits of the form $\exp\{-i\frac{\theta}{2} Z_1 Z_2\dots Z_n \}$. By using basis transformations, one can construct a circuit for any generic Pauli string. UCC factors will use the same strategy. To construct the circuit, one can start with the circuit for evaluating $\exp\{ -i\frac{\theta}{2} Z_1Z_1 \dots Z_n \}$ and then apply basis transformations to evaluate the exponential of any Pauli string. 

The circuit to evaluate $\exp\{-i\frac{\theta}{2} Z_1 Z_2\dots Z_n \}$ requires a CNOT cascade, a $U_\theta$ gate applied to the last qubit, and then a reversed CNOT cascade. The CNOT cascade calculates the parity of the circuit. After the first CNOT cascade, the last qubit in the cascade will be $\ket{0}$ if the overall parity was even, and $\ket{1}$ if the parity was odd. The $U_\theta$ gate applied on the last qubit will give a phase of $\exp\{ -i\theta\}$ if the parity is even, and a phase of $\exp\{ +i\theta\}$ if the parity is odd. The following CNOT cascade is applied to cancel out the first CNOT cascade, reverting the qubits to their original value now with a resulting overall application of an exponentiated Pauli string. Figure \ref{circuit_expz} shows an example implementation of $\exp\{-i\frac{\theta}{2} Z_1 Z_2 Z_3 Z_4 \}$. 

In order to evaluate a generic Pauli string consisting of $Z$,$X$, and $Y$, a basis transformation can be applied before the CNOT cascades such that the effective Pauli string is that of only $Z$'s. If the $i$th gate in the Pauli string is an $X$,  a Hadamard gate is sandwiched around the CNOT cascade on the $i$th qubit. This leads to the effective exponential containing a $Z$ since $HXH = Z$. Similarly, if an exponentiated $Y$ gate is applied, a $R_x(-\frac{\pi}{2})$ gate is sandwiched around the CNOT cascade. Figure \ref{ucc_x} shows an example circuit to apply $\exp\{ -i\frac{\theta}{2} Z_1Z_2Z_3X_4\}$. In this example, since the last Pauli in the exponentiated string is an $X$, a Hadamard gate is applied before and after in order to transform the basis and effectively make the circuit an exponential of $Z$s.

\begin{figure}[h]
\includegraphics[width=6cm]{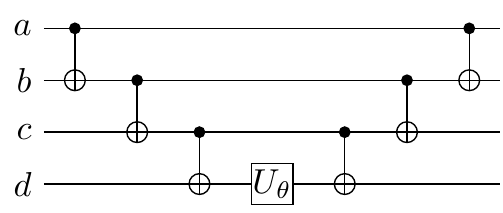}
\centering
\caption{Example of a circuit implementing $\exp\{ -i\frac{\theta}{2} Z_a\otimes Z_b\otimes Z_c \otimes Z_d\}$ for four qubits.}
\label{circuit_expz}
\end{figure}

\begin{figure}[h]
\includegraphics[width=6cm]{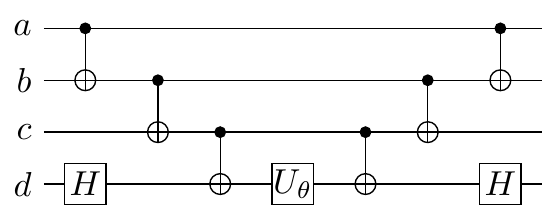}
\centering
\caption{Example of a circuit implementing $\exp\{ -i\frac{\theta}{2} Z_a\otimes Z_b \otimes Z_c \otimes X_d\}$ for four qubits. To apply the $X$ on a different qubit, Hadamard gates can be sandwiched around the respective qubits.}
\label{ucc_x}
\end{figure}

\begin{table}[h]
  \begin{tabular}{c c c c|c || c c c c|c}
         $a$ & $b$ & $c$ & $d$ & Parity & $a$ & $b$ & $c$ & $d$ & Parity\\
         \hline
         0 & 0 & 0 & 0 & 0 & 0 & 0 & 0 & 1 & 1 \\
         0 & 0 & 1 & 0 & 1 & 0 & 0 & 1 & 1 & 0 \\
         0 & 1 & 0 & 0 & 1 & 0 & 1 & 0 & 1 & 0 \\
         0 & 1 & 1 & 0 & 0 & 0 & 1 & 1 & 1 & 1 \\
         1 & 0 & 0 & 0 & 1 & 1 & 0 & 0 & 1 & 0 \\
         1 & 0 & 1 & 0 & 0 & 1 & 0 & 1 & 1 & 1 \\
         1 & 1 & 0 & 0 & 0 & 1 & 1 & 0 & 1 & 1 \\
         1 & 1 & 1 & 0 & 1 & 1 & 1 & 1 & 1 & 0 
    \end{tabular}
    \caption{The parity is the value on qubit $d$ after the CNOT cascade is applied.}
    \label{table:parity calc}
\end{table}

\begin{figure}[h]
\includegraphics[width=7cm]{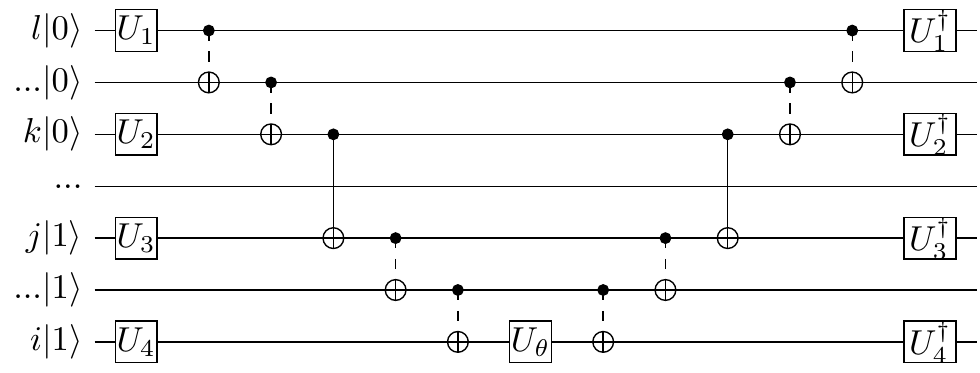}
\centering
\caption{Doubles UCC circuit as discussed in Refs. \onlinecite{barkoutsos_2018} and \onlinecite{romero_babbush_mcclean_hempel_love_aspuru-guzik_2018}. For a general doubles operator, the circuit must be applied eight times, with different combinations of $U$ gates each time. The $U$-gate choices are summarized in Table~\ref{table:circuit_u}. The dashed CNOT gates are part of a CNOT cascade.  }
\label{ucc_doubles}
\end{figure}

\begin{table}[h]
  \begin{tabular}{c|c|c|c|c}
         Subcircuit &$U_1$ & $U_2$ & $U_3$ & $U_4$ \\
         \hline
         1 & $H$ & $H$ & $R_x(-\frac{\pi}{2})$  & $H$\\
         2 & $R_x(-\frac{\pi}{2})$& $H$ & $R_x(-\frac{\pi}{2})$&$R_x(-\frac{\pi}{2})$ \\
         3& $H$ & $R_x(-\frac{\pi}{2})$& $R_x(-\frac{\pi}{2})$&$R_x(-\frac{\pi}{2})$ \\
         4 & $H$ & $H$ & $H$ &$R_x(-\frac{\pi}{2})$ \\
         5 & $R_x(-\frac{\pi}{2})$& $H$ & $H$ & $H$\\
         6 & $H$ &$R_x(-\frac{\pi}{2})$ & $H$ & $H$ \\
         7 & $R_x(-\frac{\pi}{2})$& $R_x(-\frac{\pi}{2})$&$R_x(-\frac{\pi}{2})$ & $H$ \\
         8 & $R_x(-\frac{\pi}{2})$&$R_x(-\frac{\pi}{2})$ & $H$ & $R_x(-\frac{\pi}{2})$
    \end{tabular}
    \caption{Eight different subcircuits that must be run sequentially to apply a UCC doubles factor to a wavefunction. Realizations of the generic unitary operators $U_i$ in terms of Hadamard operators and rotations of $\pi/2$ about the $x$-axis for each subcircuit used in the UCC doubles circuit in Fig.~\ref{ucc_doubles}. The $H$ gate converts the basis to the $x$-basis in order to calculate the exponential of $X$. The $R_x\left( -\frac{\pi}{2}\right)$ gate converts the basis to the $y$-basis in order to calculate the exponential of $Y$. When running an exponential of $Z$, no basis transformation is needed. Since the relevant operators all commute, the subcircuits can be run in any order, but all eight need to appear exactly once to complete the full circuit.}
    \label{table:circuit_u}
\end{table}

In applying the UCC ansatz, circuits such as Fig.~\ref{ucc_doubles} must be re-run multiple times after applying all of the $2^{2n-1}$ different basis transformations \cite{romero_babbush_mcclean_hempel_love_aspuru-guzik_2018}. A general factorized doubles UCC operator can be rewritten as Eq.~(\ref{doublefactor_jw}), and implemented exactly by the circuit shown in Fig.~\ref{ucc_doubles}.

\subsection{Control Gate Identities}
In order to implement some of the more complicated UCC factors needed for the decomposition method, we must break down the general control unitaries into standard gates. To get an accurate gate count of CNOTs, we use the method in Ref.~\cite{elementary_gates}. 

Figure \ref{singly ctrl unitary} shows the breakdown of a singly-controlled unitary. A singly-controlled unitary gate can be broken down into three single qubit gates and two CNOT gates.

\begin{figure}[h]
\includegraphics[width=6cm]{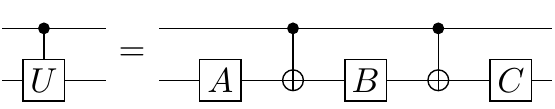}
\centering
\caption{Decomposition of a singly controlled unitary gate \cite{elementary_gates}.}
\label{singly ctrl unitary}
\end{figure}

Figure \ref{doubly ctrl unitary} shows the breakdown of a general doubly-controlled unitary. A doubly-controlled unitary can be broken down into three singly-controlled unitaries and two CNOTs. Thus, in total, a doubly-controlled unitary gate requires 8 CNOT gates and 9 unitaries. 

\begin{figure}[h]
\includegraphics[width=6cm]{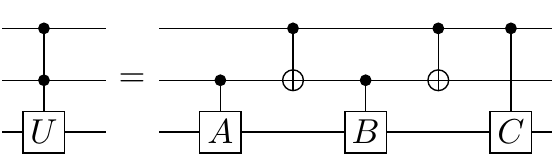}
\centering
\caption{Decomposition of a doubly controlled unitary gate \cite{elementary_gates}.}
\label{doubly ctrl unitary}
\end{figure}

\section{Decomposition Method}



We start by discussing the general schematic for the triple and quadruple excitation. We then show how one could use these schemes to generate higher-rank excitations.

The general principle for this method is as follows. In order to use mainly singles and doubles in the decomposition, we introduce ancilla orbitals. These are non-physical orbitals that act as placeholders. We effectively create higher-rank excitations by exciting these ancilla orbitals, and then applying another excitation to place them back into the correct orbitals. This is done in a way such that states that are not be affected by the higher-rank excitations will remain unaffected after the full procedure is complete. 

\begin{figure*}[t]
    \centering
    \includegraphics[width=\textwidth,height=7cm]{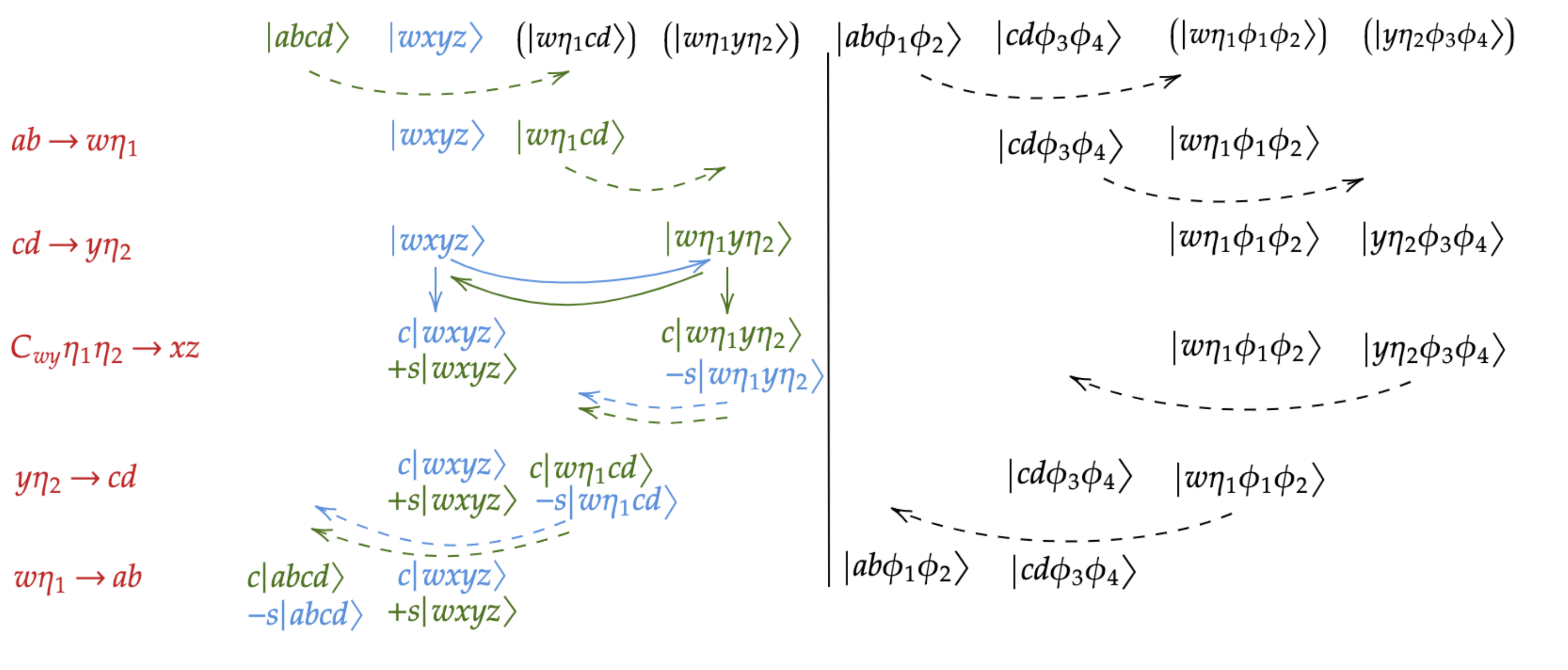}
    \caption{Diagram for the quadruples decomposition scheme. The target states of the UCC quadruples are colored in green and blue. The $4$ `shelved' states are placed on the right of the diagram. `Shelved' states are the ones that are not supposed to be affected by the quadruple operation and the decomposition procedure manages to keep them intact. The states in parentheses are virtual, serving as placeholders to better illustrate some operations. Dotted lines represent full $\pi/2$ rotations whereas solid lines represent rotations with generic angles. The coefficients $c$ and $s$ are $cosine$ and $sine$ functions of said angles from Eq. (\ref{eq:identity}).}
    \label{fig:thebigtable}
\end{figure*}

\subsection{Quadruple Excitations}
We present a schematic to create a quadruple excitation using two ancilla. Our goal is to apply an operator equivalent to $\hat{A}_{abcd}^{wxyz}$

Without loss of generality, we assume that our starting state is a general state of the following form:
\begin{equation}
    \ket{\Psi} =  \xi_1 \ket{abcd} + \xi_2 \ket{ab\phi_1 \phi_2} + \xi_3 \ket{cd \phi_3 \phi_4 } +  \xi_4 \ket{wxyz}
    \label{eq:init}
\end{equation}

Here, $a,b,c,d,w,x,y,z$ are occupied orbitals and $\phi_i$ can be any arbitrary orbital that is not $a,b,c,d,w,x,y,z$. $\xi_i$ is the coefficient associated with each state. We omit states for which the UCC factor acts like the identity and which are not touched by the operators used to construct the quadruple excitation. Note that a general state can have many terms of the form given in Eq.~(\ref{eq:init}) as a linear superposition over different $\phi_i$ with different coefficients. But, because the procedure we use is linear, those other terms will be taken care of in the circuit, so we do not need to include them explicitly in our analysis. 

We illustrate the procedure graphically in Fig.~\ref{fig:thebigtable}. A UCC quadruple operator with angle $\theta$ should transform the wavefunction $\ket{\Psi}$ as follows:
\begin{align}
    &e^{\theta \big(\hat A_{abcd}^{wxyz} - \big(\hat A_{abcd}^{wxyz}\big)^\dagger\big)}\ket{\Psi} \nonumber \\
    = &\cos\theta\xi_1\ket{abcd} + \sin\theta\xi_1\ket{wxyz} + \xi_2\ket{ab\phi_1\phi_2} \nonumber \\
    + &\xi_3\ket{cd\phi_3\phi_4} + \cos\theta\xi_4\ket{wxyz}-\sin\theta \xi_4\ket{abcd}. \label{eq:quad}
\end{align}
The change in sign of the last term arises because it is a de-excitation.

The quadruple excitation requires four doubles and one double-qubit controlled UCC double operation. Table \ref{quaddecomp} shows the operations used to create a quadruple excitation $\hat{A}_{abcd}^{wxyz}$. The leftmost column indicates what our target orbitals are and what they become. 

Starting from our initial state, we first apply a standard doubles UCC operator that transforms $ab \rightarrow w\eta_1$ with $\theta = \pi /2$. This operator will take the occupied $\ket{ab}$ and mix it with $\ket{w \eta_1}$; it does not de-excite any state, because the $\eta_1$ qubit is initially in the 0 state. For example, when applied to the state $\ket{abcd}$:
\begin{align}
    & e^{\frac{\pi}{2} \big(\hat A^{w\eta_1} _{ab} - \big(\hat A^{w\eta_1} _{ab}\big)^\dagger \big) }\xi_1\ket{abcd} \nonumber \\
    = & \cos{\frac{\pi}{2}}\xi_1\ket{abcd} + \sin{\frac{\pi}{2}}\xi_1\ket{w\eta_1 cd} = \xi_1\ket{w\eta_1 cd} 
\end{align}
After this operator is applied, the states $\ket{abcd}$ and $\ket{ab \phi_1 \phi_2}$ will be changed. Hence, after this operation, the initial state in Eq.~(\ref{eq:init}) is transformed into the following: 
\begin{equation}
    \ket{\Psi} \to \xi_1\ket{w\eta_1 cd} + \xi_2\ket{w\eta_1 \phi_1\phi_2} + \xi_3\ket{cd\phi_3\phi_4} + \xi_4\ket{wxyz}.
    \label{eq:middle1}
\end{equation}
See the top line of Fig.~\ref{fig:thebigtable} for a summary of this first step.
Next, another standard doubles UCC operator transforms $cd \rightarrow y\eta_2$ with $\theta = \pi/2$. This will change the state $\ket{w\eta_1 cd}$ to $\ket{w\eta_1y\eta_2}$ and  the state $\ket{cd\phi_3\phi_4}$ to $\ket{y\eta_2\phi_3\phi_4}$; again, there is no de-excitation because the $\eta_2$ qubit is initially in the 0 state. The resulting transformed state is: 
\begin{equation}
    \ket{\Psi}\to \xi_1\ket{w\eta_1 y \eta_2} + \xi_2\ket{w\eta_1 \phi_1\phi_2} + \xi_3\ket{y\eta_2\phi_3\phi_4} + \xi_4\ket{wxyz}. \label{eq:middle2}
\end{equation}
This result is summarized in the second line of Fig.~\ref{fig:thebigtable}.

The operator labelled $C_{wy}\eta_1 \eta_2  \rightarrow xz$ is a doubly controlled UCC double operator. If the orbitals $wy$ are present, then we take $\eta_1 \eta_2$ and apply the UCC operator to take it to a linear superposition of $x z$ and $\eta_1\eta_2$. The operator has a general angle $\theta$ and yields
\begin{align}
    & e^{\theta \big(\hat A^{xz} _{\eta_1 \eta_2} - \big(\hat A^{xz} _{\eta_1 \eta_2}\big)^\dagger \big) }(\xi_1\ket{w\eta_1 y\eta_2}+\xi_4\ket{wxyz}) \nonumber \\
    = & \cos{\theta}\xi_1\ket{w\eta_1 y\eta_2} + \sin{\theta}\xi_1\ket{wxyz} \nonumber \\
    + & \cos{\theta}\xi_4\ket{wxyz} - \sin{\theta}\xi_4\ket{w\eta_1 y\eta_2} 
\end{align}
when acting on the two states that are transformed by it.
The negative sign arises because that term is a de-excitation.  The result after this step is: 
\begin{align}
    \ket{\Psi}\to &\cos\theta\xi_1\ket{w\eta_1 y\eta_2} + \sin\theta\xi_1\ket{wxyz}\nonumber \\
     + &\xi_2\ket{w\eta_1 \phi_1\phi_2} + \xi_3\ket{y\eta_2\phi_3\phi_4} \nonumber\\
    + &\cos\theta\xi_4\ket{wxyz}-\sin\theta \xi_4\ket{w\eta_1 y\eta_2}.\label{eq:middle3}
\end{align}
We have added in trigonometric factors, which multiply whatever the original coefficients were. This operation is depicted in the third line of Fig.~\ref{fig:thebigtable}.

The next two doubles act as corrections. They will remove the ancilla orbitals from the states. The double UCC that takes $y\eta_2\rightarrow cd$ with $\theta = \pi/2$ changes $\ket{y\eta_2 \phi_3 \phi_4}$ to $\ket{cd\phi_3\phi_4}$ and $-\ket{w\eta_1y\eta_2}\rightarrow -\ket{w\eta_1cd}$. The state after this step is:
\begin{align}
    \ket{\Psi}\to &\cos\theta\xi_1\ket{w\eta_1 cd} + \sin\theta\xi_1\ket{wxyz}\nonumber \\
     + &\xi_2\ket{w\eta_1 \phi_1\phi_2} + \xi_3\ket{cd\phi_3\phi_4} \nonumber\\
    + &\cos\theta\xi_4\ket{wxyz}-\sin\theta \xi_4\ket{w\eta_1 cd}, \label{eq:middle4}
\end{align}
see the second to last line of Fig.~\ref{fig:thebigtable}.
Finally, the very last double takes $w\eta_1 \rightarrow ab$ with $\theta = \pi/2$. This takes the state $\ket{w\eta_1 \phi_1\phi_2}$ to $\ket{ab\phi_1\phi_2}$ and $-\ket{w\eta_1 cd} $ to $-\ket{abcd}$. The final state is therefore: 
\begin{align}
    \ket{\Psi}\to &\cos\theta\xi_1\ket{abcd} + \sin\theta\xi_1\ket{wxyz}\nonumber \\
     + &\xi_2\ket{ab\phi_1\phi_2} + \xi_3\ket{cd\phi_3\phi_4} \nonumber\\
    + &\cos\theta\xi_4\ket{wxyz}-\sin\theta \xi_4\ket{abcd}, \label{eq:final}
\end{align}
which is identical to our goal Eq.(\ref{eq:quad}); see the last line of Fig.~\ref{fig:thebigtable} for more detail.

We went through this derivation assuming there was only one term of the form $|ab\phi_1\phi_2\rangle$ in the expansion. But, of course, there can be many such terms. However, since this term gets ``shelved'' to a state that sits out of all of the remaining UCC terms except for the last one, it should be clear that adding additional terms of this form, simply shelves those additional terms (in linear superposition) and then brings them back. So, this approach works for an arbitrary linear combination of terms of the form $|ab\phi_1\phi_2\rangle$. A similar conclusion can be reached for the terms of the form $|cd\phi_3\phi_4\rangle$ (with them being brought back in the second to last step).

\begin{table}[h]
    \resizebox{\columnwidth}{!}{
  \begin{tabular}{|c|l|l|l|l|}
  \hline
  State       & $\ket{abcd}$ & $\ket{ab\phi_1 \phi_2}$ & $\ket{cd\phi_3 \phi_4}$ & $ \ket{wxyz}$
         \\
         \hline
         
         $ab \rightarrow w\eta_1$&$\ket{w\eta_1 cd}$&$\ket{w\eta_1 \phi_1 \phi_2}$&$\ket{cd\phi_3 \phi_4}$&$ \ket{wxyz}$ \\ \hline
         $cd \rightarrow y\eta_2$&$\ket{w\eta_1 y\eta_2}$&$\ket{w\eta_1 \phi_1 \phi_2}$&$\ket{y\eta_2 \phi_3 \phi_4}$&$ \ket{wxyz}$ \\ \hline
         $C_{wy}\eta_1 \eta_2  \rightarrow xz$  &\begin{tabular}{@{}l@{}} $c\ket{w\eta_1y\eta_2}$\\$+s\ket{wxyz}$ \end{tabular}&$\ket{w\eta_1 \phi_1 \phi_2}$&$\ket{y\eta_2 \phi_3 \phi_4}$&\begin{tabular}{@{}l@{}} $c\ket{wxyz}$\\$-s\ket{w\eta_1y\eta_2}$ \end{tabular} \\ \hline
         $y\eta_2 \rightarrow cd$ &\begin{tabular}{@{}l@{}} $c\ket{w\eta_1cd}$\\$+s\ket{wxyz}$ \end{tabular}&$\ket{w\eta_1 \phi_1 \phi_2}$&$\ket{cd \phi_3 \phi_4}$&\begin{tabular}{@{}l@{}} $c\ket{wxyz}$\\$-s\ket{w\eta_1 cd}$ \end{tabular} \\ \hline
         $w\eta_1 \rightarrow ab$ &\begin{tabular}{@{}l@{}} $c\ket{abcd}$\\$+s\ket{wxyz}$ \end{tabular}&$\ket{ab \phi_1 \phi_2}$&$\ket{cd \phi_3 \phi_4}$& \begin{tabular}{@{}l@{}} $c\ket{wxyz}$\\$-s\ket{abcd}$ \end{tabular} \\ \hline
    \end{tabular}
    }
    \caption{Schematic for the quadruple decomposition algorithm. The $c$ and $s$ in the second and last columns represent $\cos\theta$ and $\sin\theta$ respectively. The angles with which the four UCC doubles apply to the wavefunctions are all $\pi/2$, whereas the angle of the doubly controlled UCC double operator is a generic one $\theta$.}
    \label{quaddecomp}
\end{table}

Next we will show that each and every step of the algorithm is necessary to successfully decompose a UCC quadruple operator. 

One might assume that it is possible to break down the quad with two doubles. For example, naively applying a double that takes $ab\rightarrow wx$ and $cd \rightarrow yz$ would take $\ket{abcd} \rightarrow \ket{wxyz}$. This approach will fail even if only $\ket{abcd}$ or $\ket{wxyz}$ are present in the wavefunction. Suppose we have an initial wavefunction $\ket{\Psi}=\xi_1\ket{abcd}+\xi_2\ket{wxyz}$, where $\xi_1 ^2 + \xi_2 ^2 = 1$ and $\xi_1, \xi_2\in \mathbb{R}$. The first step $ab\rightarrow wx$ acting on the wavefunction $\ket{\Psi}$ yields
\begin{align}
    \ket{\Psi}\to &\cos\theta\xi_1\ket{abcd}+\sin\theta\xi_1\ket{wxcd} \nonumber \\
    &\cos\theta\xi_2\ket{wxyz}-\sin\theta\xi_2\ket{abyz}.
\end{align}
The second step $cd \rightarrow yz$ yields
\begin{align}
    \ket{\Psi}\to &\cos^2\theta\xi_1\ket{abcd} +\sin\theta\cos\theta\xi_1\ket{abyz}\nonumber\\
    +&\cos\theta\sin\theta\xi_1\ket{wxcd}+\sin^2\theta\xi_1\ket{wxyz} \nonumber \\
    +&\cos^2\theta\xi_2\ket{wxyz} - \cos\theta\sin\theta\xi_2\ket{wxcd}\nonumber\\
    -&\cos\theta\sin\theta\xi_2\ket{abyz} + \sin^2\theta\xi_2\ket{abcd}.
\end{align}
Recall the goal here is to replicate the operation
\begin{align}
    \ket{\Psi} \to &\cos\theta\xi_1\ket{abcd}+\sin\theta\xi_1\ket{wxyz} \nonumber \\
    +&\cos\theta\xi_2\ket{wxyz}-\sin\theta\xi_2\ket{abcd},
\end{align}
which the naive method fails miserably. 

We introduce the ancilla orbitals to circumvent such an issue. First let us examine the scheme as shown in Table~\ref{tab:nocontrol}.
\begin{table}[h]
    \centering
    \begin{tabular}{|c|c|}
    \hline
        Step & Operation \\ \hline
        1 & $ab\rightarrow w\eta_1$ \\ \hline
        2 & $cd\rightarrow y\eta_2$ \\ \hline
        3 & $\eta_1\eta_2 \rightarrow xz$ \\ \hline
        4 & $y\eta_2\rightarrow cd$ \\ \hline
        5 & $w\eta_1\rightarrow ab$ \\ \hline
    \end{tabular}
    \caption{A seemingly working scheme that trys to decompose the UCC quadruple operator with the aid of two ancilla qubits $\eta_1$ and $\eta_2$. Steps 1, 2, 4, and 5 are associated with angle $\theta=\pi/2$. The angle used in step 3 is arbitrary.} 
    \label{tab:nocontrol}
\end{table}

Although sometimes successful at delivering the correct resulting wavefunctions, this method breaks down if states $\ket{xz\phi_1\phi_2}$ are present where $\phi_i$ are arbitrary orbitals. For example, assume we have a wavefunction $\ket{\Psi_{tr}} = \xi_1\ket{abcd} + \xi_2\ket{acxz}+\xi_3\ket{abyx}+\xi_4\ket{cdwz}+\xi_5\ket{wxyz}$, the intermediate states obtained from using the scheme presented in Tab.~\ref{tab:nocontrol} are shown in Tab.~\ref{tab:nocontroldecomp}.

\begin{table}[h!]
    \resizebox{\columnwidth}{!}{
  \begin{tabular}{|c|l|l|l|l|l|}
  \hline
  State       & $\ket{abcd}$ & $\ket{acxz}$ & $\ket{abyx}$ & $\ket{cdwz}$ & $\ket{wxyz}$
         \\
         \hline
         
         $ab \rightarrow w\eta_1$&$\ket{w\eta_1 cd}$&$\ket{acxz}$&$\ket{w\eta_1 yx}$&$\ket{cdwy}$&$\ket{wxyz}$ \\ \hline
         $cd \rightarrow y\eta_2$&$\ket{w\eta_1 y\eta_2}$&$\ket{acxz}$&$\ket{w\eta_1 yx}$&$\ket{y\eta_2 wz}$&$\ket{wxyz}$ \\ \hline
         $\eta_1 \eta_2  \rightarrow xz$  &\begin{tabular}{@{}l@{}} $c\ket{w\eta_1y\eta_2}$\\$+s\ket{wxyz}$ \end{tabular}&\begin{tabular}{@{}l@{}} $c\ket{acxz}$\\$-s\ket{ac\eta_1 \eta_2}$ \end{tabular}&$\ket{w\eta_1 yx}$&$\ket{y\eta_2 wz}$&\begin{tabular}{@{}l@{}} $c\ket{wxyz}$\\$-s\ket{w\eta_1y\eta_2}$ \end{tabular} \\ \hline
         $y\eta_2 \rightarrow cd$ &\begin{tabular}{@{}l@{}} $c\ket{w\eta_1cd}$\\$+s\ket{wxyz}$ \end{tabular}&\begin{tabular}{@{}l@{}} $c\ket{acxz}$\\$-s\ket{ac\eta_1 \eta_2}$ \end{tabular}&$\ket{w\eta_1 yx}$&$\ket{cdwz}$&\begin{tabular}{@{}l@{}} $c\ket{wxyz}$\\$-s\ket{w\eta_1 cd}$ \end{tabular}\\ \hline
         $w\eta_1 \rightarrow ab$ &\begin{tabular}{@{}l@{}} $c\ket{abcd}$\\$+s\ket{wxyz}$ \end{tabular}&\begin{tabular}{@{}l@{}} $c\ket{acxz}$\\$-s\ket{ac\eta_1 \eta_2}$ \end{tabular}&$\ket{abyx}$&$\ket{cdwz}$&\begin{tabular}{@{}l@{}} $c\ket{wxyz}$\\$-s\ket{abcd}$ \end{tabular} \\ \hline
    \end{tabular}
    }
    \caption{States of the wavefunction $\ket{\Psi_{tr}}$ transformed by operators from Tab.~\ref{tab:nocontrol}. It is noticable here that the state $\ket{acxz}$ will be affected by the critical step $\eta_1 \eta_2 \rightarrow xz$ due to the fact that the UCC excitation operator is also a UCC de-excitation operator and the resulting state will not be corrected back into $\ket{acxz}$ either. Therefore the scheme shown in Tab.~\ref{tab:nocontrol} fails when $\ket{xz\phi_1 \phi_2}$ is present. }
    \label{tab:nocontroldecomp}
\end{table}
Hence the usage of a doubly controlled UCC double operation with the two control qubits being placed onto the orbitals $w$ and $y$ to make sure that only the state $\ket{wxyz}$ will be affected by the double $\eta_1 \eta_2 \rightarrow xz$.

\subsection{Other Rank Excitations}

Like the quadruple excitation, the triple excitation involves five operations. We follow a similar architecture for the triples as we do for the quadruple excitations. It involves two doubles, two singles, and one singly controlled double.

Table \ref{tripledecomp} summarizes the operations needed to apply the triple $\hat{A}_{abc}^{wxy}$. Note that a traditional way of implementing a UCC triple operator uses less two-qubit gates than this method for $N \leq 18$, however the gate count for CNOTs present in the traditional circuit will quickly outnumber that in our circuit. Another direction to approach the triples is to use Givens rotations together with control gates and swap gates \cite{xanadu_2021}. However for large systems consisting of a large number of active orbitals, multi-qubit controlled swaps and multi-qubit Givens operators will quickly become inefficient.

\begin{table}[h]
    \resizebox{\columnwidth}{!}{
  \begin{tabular}{|c|l|l|l|l|}
  \hline
  State       & $\ket{abc}$ & $\ket{ab\phi_1}$ & $\ket{cd\phi_2 \phi_3}$ & $ \ket{wxy}$
         \\
         \hline
         
         $ab \rightarrow w\eta_1$&$\ket{w\eta_1 c}$&$\ket{w\eta_1 \phi_1}$&$\ket{c\phi_2 \phi_3}$&$ \ket{wxy}$ \\ \hline
         $c \rightarrow \eta_2$&$\ket{w\eta_1 \eta_2}$&$\ket{w\eta_1 \phi_1}$&$\ket{\eta_2 \phi_2 \phi_3}$&$ \ket{wxy}$ \\ \hline
         $C_{w}\eta_1 \eta_2  \rightarrow xy$  &\begin{tabular}{@{}l@{}} $c\ket{w\eta_1\eta_2}$\\$+s\ket{wxy}$ \end{tabular}&$\ket{w\eta_1 \phi_1}$&$\ket{\eta_2 \phi_3}$&\begin{tabular}{@{}l@{}} $c\ket{wxy}$\\$-s\ket{w\eta_1\eta_2}$ \end{tabular} \\ \hline
         $\eta_2 \rightarrow c$ &\begin{tabular}{@{}l@{}} $c\ket{w\eta_1c}$\\$+s\ket{wxy}$ \end{tabular}&$\ket{w\eta_1 \phi_1}$&$\ket{c \phi_2 \phi_3}$&\begin{tabular}{@{}l@{}} $c\ket{wxy}$\\$-s\ket{w\eta_1 c}$ \end{tabular} \\ \hline
         $w\eta_1 \rightarrow ab$ &\begin{tabular}{@{}l@{}} $c\ket{abc}$\\$+s\ket{wxy}$ \end{tabular}&$\ket{ab \phi_1}$&$\ket{c \phi_2 \phi_3}$& \begin{tabular}{@{}l@{}} $c\ket{wxy}$\\$-s\ket{abc}$ \end{tabular} \\ \hline
    \end{tabular}
    }
    \caption{Schematic for the triple decomposition algorithm. The $c$ and $s$ in the second and last columns represent $\cos\theta$ and $\sin\theta$ respectively. The angles with which the four UCC doubles apply to the wavefunctions are all $\pi/2$, whereas the angle of the doubly controlled UCC double operator is  $\theta$.}
    \label{tripledecomp}
\end{table}

For higher-rank excitations, we present various methods for decomposing $N$-rank excitations in terms of lower-rank excitations. While multiple methods to break down higher-rank excitations are possible, every method will follow the same methodology. We start with two excitations into ancilla orbitals, followed by a controlled operation and then two more excitations to undo the rotation into the ancilla orbitals. In total, the process takes five operations. For an $n$-tuple excitation operator, the outer excitations should add up to $n$. For example, for a sextuple excitation, one should use a double and quad, two triples, or a single and a quintuple excitation. 

The method of choice should depend on the hardware in use, as different methods utilize different numbers of CNOTs and rotations. For example, consider the case of the the sextuplet excitation. We can either perform this with 2 doubles, 2 quads, and one quadruply controlled double, or we can use 4 triples and one quadruply controlled double. The choice  to pick is based on hardware limitations, as the gate count for different types of gates varies for these two schemes. Tables \ref{24sextdecomp} and \ref{33sextdecomp} show these two schematics. 

\begin{table}[h]
    \resizebox{\columnwidth}{!}{
  \begin{tabular}{|c|l|l|l|l|}
  \hline
  State       & $\ket{abcde}$ & $\ket{ab\phi_1 \phi_2 \phi_3}$ & $\ket{cde\phi_4 \phi_5}$ & $ \ket{vwxyz}$
         \\
         \hline
         
         $ab \rightarrow v\eta_1$&$\ket{v\eta_1 cde}$&$\ket{v\eta_1 \phi_1 \phi_2 \phi_3}$&$\ket{cde\phi_4 \phi_5}$&$ \ket{vwxyz}$ \\ \hline
         $cde \rightarrow xy\eta_2$&$\ket{v\eta_1 xy\eta_2}$&$\ket{v\eta_1 \phi_1 \phi_2 \phi_3}$&$\ket{xy\eta_2 \phi_4 \phi_5}$&$ \ket{vwxyz}$ \\ \hline
         $C_{vxy}\eta_1 \eta_2  \rightarrow wz$  &\begin{tabular}{@{}l@{}} $c\ket{v\eta_1 xy\eta_2}$\\$+s\ket{vwxyz}$ \end{tabular}&$\ket{v\eta_1 \phi_1 \phi_2 \phi_3}$&$\ket{xy\eta_2 \phi_4 \phi_5}$&\begin{tabular}{@{}l@{}} $c\ket{vwxyz}$\\$-s\ket{v\eta_1 xy\eta_2}$ \end{tabular} \\ \hline
         $xy\eta_2 \rightarrow cde$ &\begin{tabular}{@{}l@{}} $c\ket{v\eta_1cde}$\\$+s\ket{vwxyz}$ \end{tabular}&$\ket{v\eta_1 \phi_1 \phi_2\phi_3}$&$\ket{cde \phi_4 \phi_5}$&\begin{tabular}{@{}l@{}} $c\ket{vwxyz}$\\$-s\ket{v\eta_1 cde}$ \end{tabular} \\ \hline
         $v\eta_1 \rightarrow ab$ &\begin{tabular}{@{}l@{}} $c\ket{abcde}$\\$+s\ket{vwxyz}$ \end{tabular}&$\ket{ab \phi_1 \phi_2 \phi_3}$&$\ket{cde \phi_4 \phi_5}$& \begin{tabular}{@{}l@{}} $c\ket{vwxyz}$\\$-s\ket{abcde}$ \end{tabular} \\ \hline
    \end{tabular}
    }
    \caption{Schematic for the quintuple decomposition algorithm. The $c$ and $s$ in the second and last columns represent $\cos\theta$ and $\sin\theta$ respectively. The angles with which the four UCC doubles apply to the wavefunctions are all $\pi/2$, whereas the angle of the doubly controlled UCC double operator is  $\theta$.}
    \label{quintdecomp}
\end{table}

\begin{table}[h]
    \resizebox{\columnwidth}{!}{
  \begin{tabular}{|c|l|l|l|l|}
  \hline
  State       & $\ket{abcdef}$ & $\ket{ab\phi_1 \phi_2 \phi_3\phi_4}$ & $\ket{cdef\phi_5 \phi_6}$ & $ \ket{uvwxyz}$
         \\
         \hline
         
         $ab \rightarrow u\eta_1$&$\ket{u\eta_1 cdef}$&$\ket{u\eta_1 \phi_1 \phi_2 \phi_3 \phi_4}$&$\ket{cdef\phi_5 \phi_6}$&$ \ket{uvwxyz}$ \\ \hline
         $cdef \rightarrow wxy\eta_2$&$\ket{u\eta_1 wxy\eta_2}$&$\ket{u\eta_1 \phi_1 \phi_2 \phi_3 \phi_4}$&$\ket{wxy\eta_2 \phi_5 \phi_6}$&$ \ket{uvwxyz}$ \\ \hline
         $C_{uwxy}\eta_1 \eta_2  \rightarrow vz$  &\begin{tabular}{@{}l@{}} $c\ket{u\eta_1 wxy\eta_2}$\\$+s\ket{uvwxyz}$ \end{tabular}&$\ket{u\eta_1 \phi_1 \phi_2 \phi_3\phi_4}$&$\ket{wxy\eta_2 \phi_5 \phi_6}$&\begin{tabular}{@{}l@{}} $c\ket{uvwxyz}$\\$-s\ket{u\eta_1 wxy\eta_2}$ \end{tabular} \\ \hline
         $wxy\eta_2 \rightarrow cdef$ &\begin{tabular}{@{}l@{}} $c\ket{u\eta_1cdef}$\\$+s\ket{uvwxyz}$ \end{tabular}&$\ket{u\eta_1 \phi_1 \phi_2\phi_3\phi_4}$&$\ket{cdef \phi_5 \phi_6}$&\begin{tabular}{@{}l@{}} $c\ket{uvwxyz}$\\$-s\ket{u\eta_1 cdef}$ \end{tabular} \\ \hline
         $u\eta_1 \rightarrow ab$ &\begin{tabular}{@{}l@{}} $c\ket{abcdef}$\\$+s\ket{uvwxyz}$ \end{tabular}&$\ket{ab \phi_1 \phi_2 \phi_3\phi_4}$&$\ket{cdef \phi_5 \phi_6}$& \begin{tabular}{@{}l@{}} $c\ket{uvwxyz}$\\$-s\ket{abcdef}$ \end{tabular} \\ \hline
    \end{tabular}
    }
    \caption{Schematic for the 2-4 sextuple decomposition algorithm. The $c$ and $s$ in the second and last columns represent $\cos\theta$ and $\sin\theta$ respectively. The angles with which the four UCC doubles apply to the wavefunctions are all $\pi/2$, whereas the angle of the doubly controlled UCC double operator is  $\theta$.}
    \label{24sextdecomp}
\end{table}

\begin{table}[h]
    \resizebox{\columnwidth}{!}{
  \begin{tabular}{|c|l|l|l|l|}
  \hline
  State       & $\ket{abcdef}$ & $\ket{abc\phi_1 \phi_2 \phi_3}$ & $\ket{def\phi_4\phi_5 \phi_6}$ & $ \ket{uvwxyz}$
         \\
         \hline
         
         $abc \rightarrow uv\eta_1$&$\ket{uv\eta_1 def}$&$\ket{uv\eta_1 \phi_1 \phi_2 \phi_3}$&$\ket{def\phi_4\phi_5 \phi_6}$&$ \ket{uvwxyz}$ \\ \hline
         $def \rightarrow xy\eta_2$&$\ket{uv\eta_1 xy\eta_2}$&$\ket{uv\eta_1 \phi_1 \phi_2 \phi_3}$&$\ket{xy\eta_2 \phi_4 \phi_5 \phi_6}$&$ \ket{uvwxyz}$ \\ \hline
         $C_{uvxy}\eta_1 \eta_2  \rightarrow wz$  &\begin{tabular}{@{}l@{}} $c\ket{uv\eta_1 xy\eta_2}$\\$+s\ket{uvwxyz}$ \end{tabular}&$\ket{uv\eta_1 \phi_1 \phi_2 \phi_3}$&$\ket{xy\eta_2 \phi_4 \phi_5 \phi_6}$&\begin{tabular}{@{}l@{}} $c\ket{uvwxyz}$\\$-s\ket{uv\eta_1 xy\eta_2}$ \end{tabular} \\ \hline
         $xy\eta_2 \rightarrow def$ &\begin{tabular}{@{}l@{}} $c\ket{uv\eta_1 def}$\\$+s\ket{uvwxyz}$ \end{tabular}&$\ket{uv\eta_1 \phi_1 \phi_2 \phi_3}$&$\ket{def \phi_4 \phi_5 \phi_6}$&\begin{tabular}{@{}l@{}} $c\ket{uvwxyz}$\\$-s\ket{uv\eta_1 def}$ \end{tabular} \\ \hline
         $uv\eta_1 \rightarrow abc$ &\begin{tabular}{@{}l@{}} $c\ket{abcdef}$\\$+s\ket{uvwxyz}$ \end{tabular}&$\ket{abc \phi_1 \phi_2 \phi_3}$&$\ket{def \phi_4 \phi_5 \phi_6}$& \begin{tabular}{@{}l@{}} $c\ket{uvwxyz}$\\$-s\ket{abcdef}$ \end{tabular} \\ \hline
    \end{tabular}
    }
    \caption{Schematic for the 3-3 sextuple decomposition algorithm. The $c$ and $s$ in the second and last columns represent $\cos\theta$ and $\sin\theta$ respectively. The angles with which the four UCC doubles apply to the wavefunctions are all $\pi/2$, whereas the angle of the doubly controlled UCC double operator is $\theta$.}
    \label{33sextdecomp}
\end{table}

\subsection{Code for Controlled UCC Factors}
The decomposition method relies on controlled UCC factors. For example, in the quadruple excitation, we require a doubly controlled UCC factor that applies $\eta_1 \eta_2 \rightarrow xz$ with $wy$ as the control qubits. Figure \ref{doubly controlled double} shows the circuit for a doubly controlled UCC factor; note that the circuit decomposition requires an additional two ancilla qubits denoted $\alpha_1$ and $\alpha_2$. The CNOTs applied before the unitary gates are used to encode the information of the control qubits into two ancilla qubits. This way, even after the cascade is applied to account for the parity, the qubits $\ket{\alpha_1}$ and $\ket{\alpha_2}$ contain the information from the control qubits. These qubits are not involved in the CNOT cascade, but are used as control qubits for the doubly controlled rotation gate that is applied within the UCC factor. 

\begin{figure}[h]
\includegraphics[width=6cm]{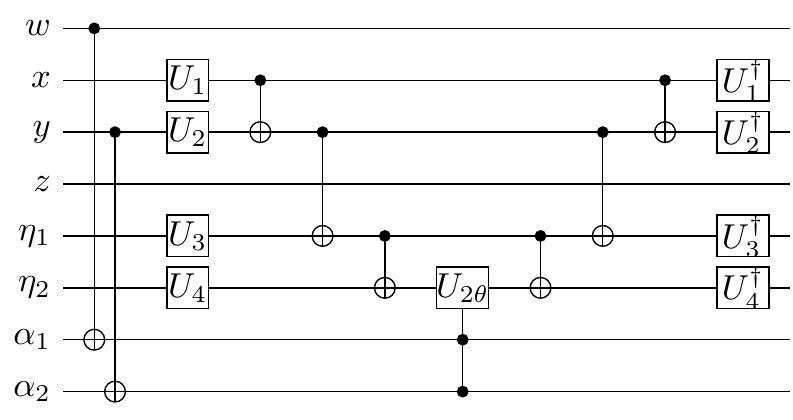}
\centering
\caption{Example of a circuit used to implement a doubly controlled double. This circuit corresponds to $C_{wy}\eta_1 \eta_2  \rightarrow xz$. The qubits $\ket{\alpha_1}$ and $\ket{\alpha_2}$ are ancilla qubits used to keep the information of $w$ and $y$ since the CNOT cascade encoding parity will alter the values.}
\label{doubly controlled double}
\end{figure}

The doubly controlled rotation gate can be broken down into unitary gates and CNOTs \cite{elementary_gates}. A doubly controlled unitary gate can be broken down into three singly controlled unitary gates and two CNOT gates. The singly controlled unitary gates can be broken down into two CNOTs and two unitary gates. Thus, in total, the doubly controlled unitary gates consist of 8 CNOTs and 6 unitaries.

\section{Results}

The benefit of the decomposition method is that the total number of CNOT gates is significantly lower than what is used in a traditional $N$-rank UCC factor. A traditional $N$-rank UCC factor with $M$ orbitals requires at most $2^{2N-1}$ single qubit rotations, $2^{2N}(M-1)$ CNOTs, and $4N(2^{2N-1})$ single qubit non-rotation gates \cite{barkoutsos_2018, romero_babbush_mcclean_hempel_love_aspuru-guzik_2018}. This count comes from assuming that there are no simplifications in the Jordan-Wigner strings.  Each UCC factor consists of a circuit that must be run $2^{2N-1}$ times for an $N$-tuple excitation. For each run, a CNOT must be applied between every neighboring set of orbitals twice, resulting in a total of $2(M-1)$ CNOTs per run. In reality, the number of CNOTs may be reduced due to simplifications in the Jordan-Wigner strings. Although this estimate for CNOT gates is generally an overestimate, the decomposition method presented above is significantly lower in gate count. 

The number of CNOT counts can also be lower if one uses a different encoding than the Jordan-Wigner encoding \cite{aspuru-guzik_uccreview}. We do not examine this strategy in detail here, primarily because such a decoding can be used for the different operators in the decomposition as well, and we anticipate similar gains in efficiency.

For comparison, consider the requirements for a quadruple excitation. A traditional quad requires $128$ single qubit rotations, $256(M-1)$ CNOTs, and $2048$ single qubit Clifford gates. The circuit used consists of $2(M-1)$ CNOTs from the cascade, one single qubit rotation applied within the cascade, and 16 single qubit gates for the basis transformations and inverse transformations. This circuit must be run 128 times. 

Our decomposition instead requires two (plus two) ancilla orbitals and is built from four doubles and one doubly controlled double. The number of required qubits will increase from $M$ to $M+4$. Two qubits are used as ancilla, and two additional qubits are needed for the controlled gate implementation. Since the decomposed quad is constructed from four doubles excitations and 1 controlled UCC doubles, the resulting CNOT count in the worst case is $4\cdot 2^4((M+2)-1) + [2^4((M+2)-1) + 2^4\cdot 8] = 80M + 208$. The first term is the CNOT count for the four standard doubles used, and the term in the brackets is the count for the controlled UCC doubles. This count increases with the number of orbitals because in each of the doubles, adding an extra orbital will add two more CNOTs into the CNOT cascade that calculates parity.
The $2^4\cdot 8$ CNOT gates comes from breaking down the controlled rotation gate \cite{elementary_gates}. Note that, although we have a total of $M+4$ total qubits, two of the qubits are not involved in the CNOT cascade.  Compared to that of a standard quadruple, the order is much less in the worst case count. 

Similarly, for other higher-rank UCC excitations, the CNOT count of the decomposition method is much lower. Figures 5-8 show the worst case gate counts for the decomposition method against the method proposed in \cite{barkoutsos_2018, romero_babbush_mcclean_hempel_love_aspuru-guzik_2018}. 

In the NISQ era, optimizing the ansatz for current hardware is necessary. In the near term, circuits that reduce circuit depth and number of CNOTs in exchange for a few additional qubits can be highly beneficial. 

\begin{figure}[htb]
    \centering
    \includegraphics[width=0.45\textwidth]{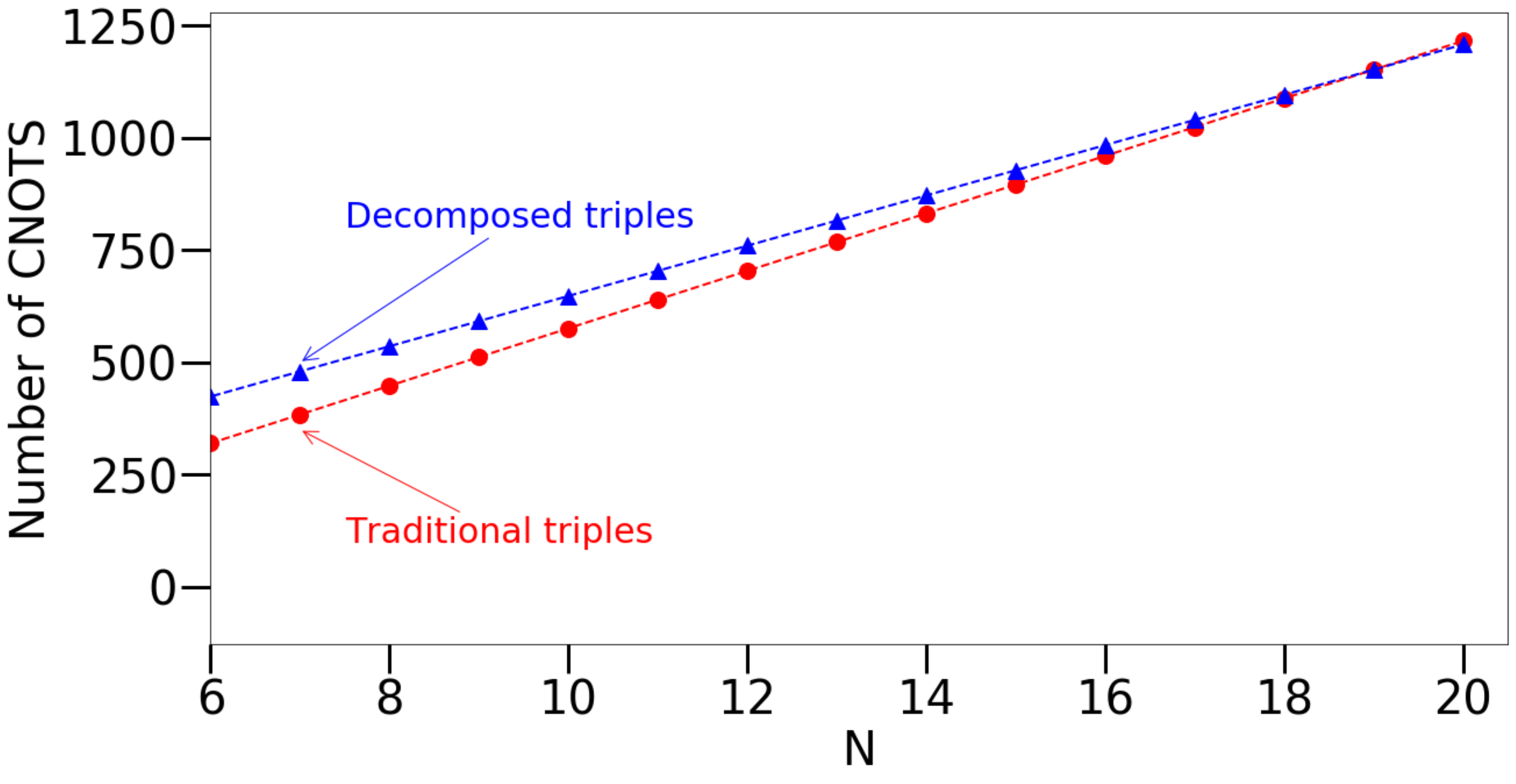}
    \caption{CNOT gate counts of traditional triples and decomposed triples.}
\end{figure}

\begin{figure}[htb]
    \centering
    \includegraphics[width=0.45\textwidth]{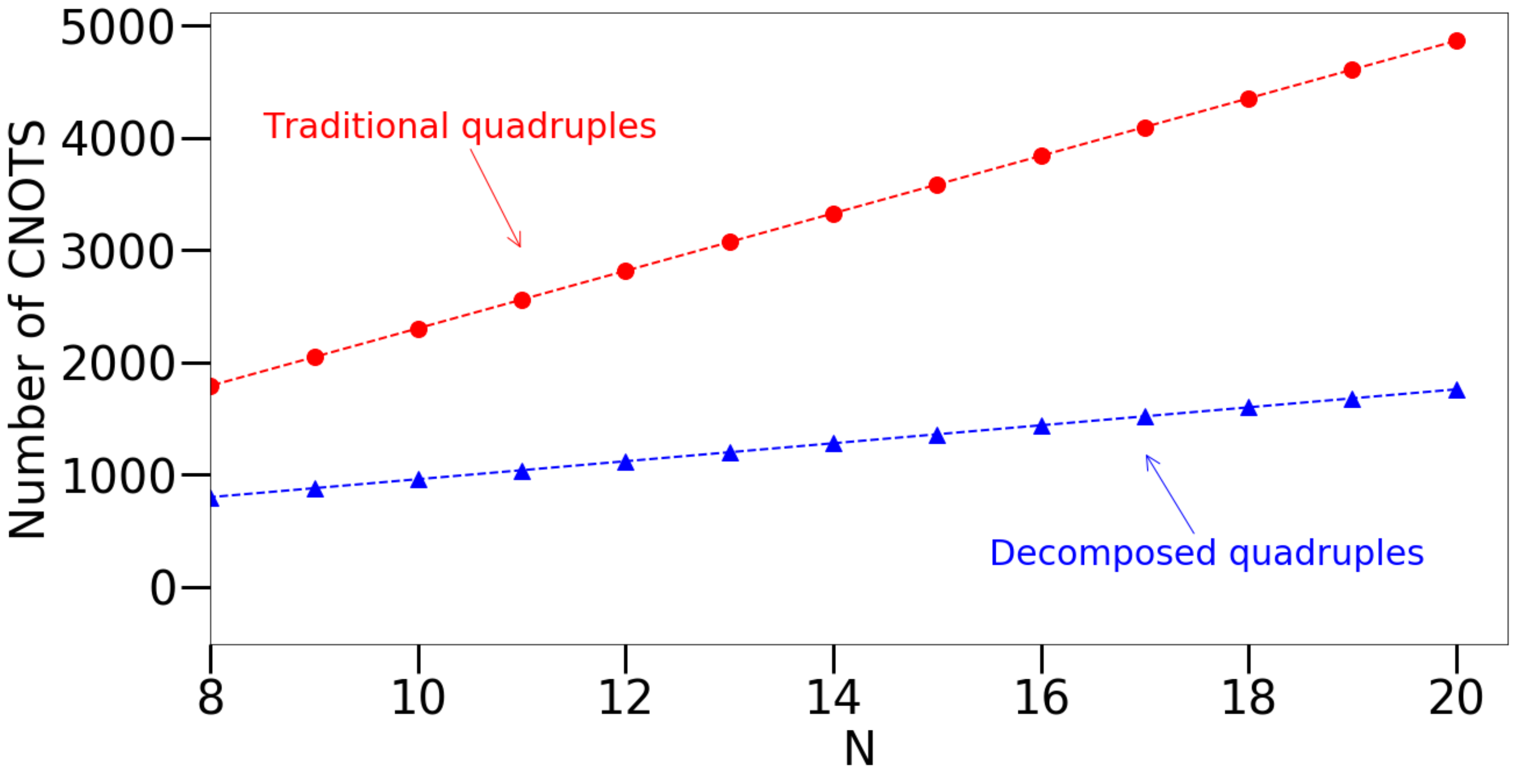}
    \caption{CNOT gate counts of traditional quadruples and decomposed quadruples.}
\end{figure}

\begin{figure}[htb]
    \centering
    \includegraphics[width=0.45\textwidth]{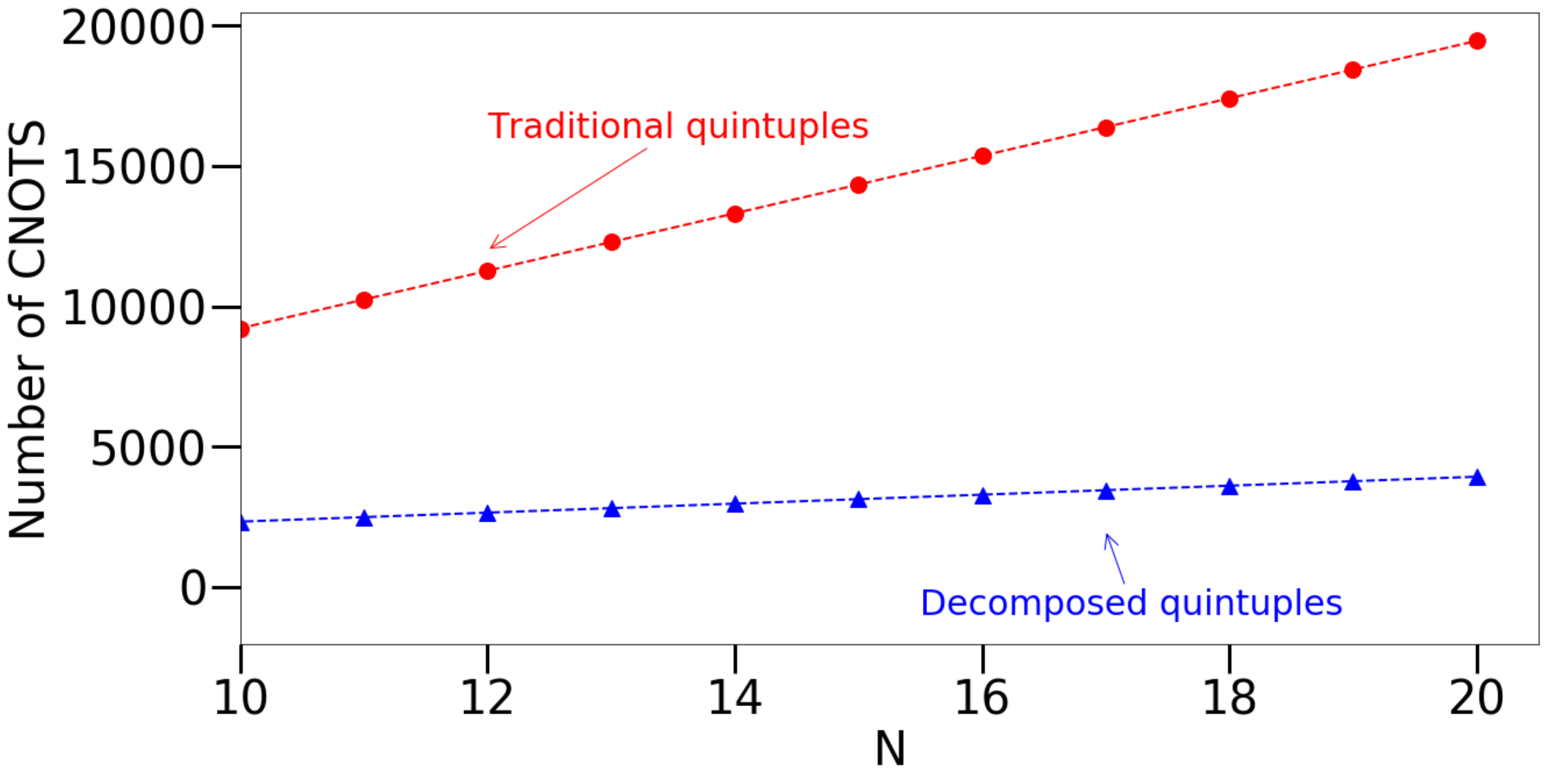}
    \caption{CNOT gate counts of traditional quintuples and decomposed quintuples.}
\end{figure}

\begin{figure}[htb]
    \centering
    \begin{subfigure}{0.45\textwidth}
        \centering
        \includegraphics[width=\textwidth]{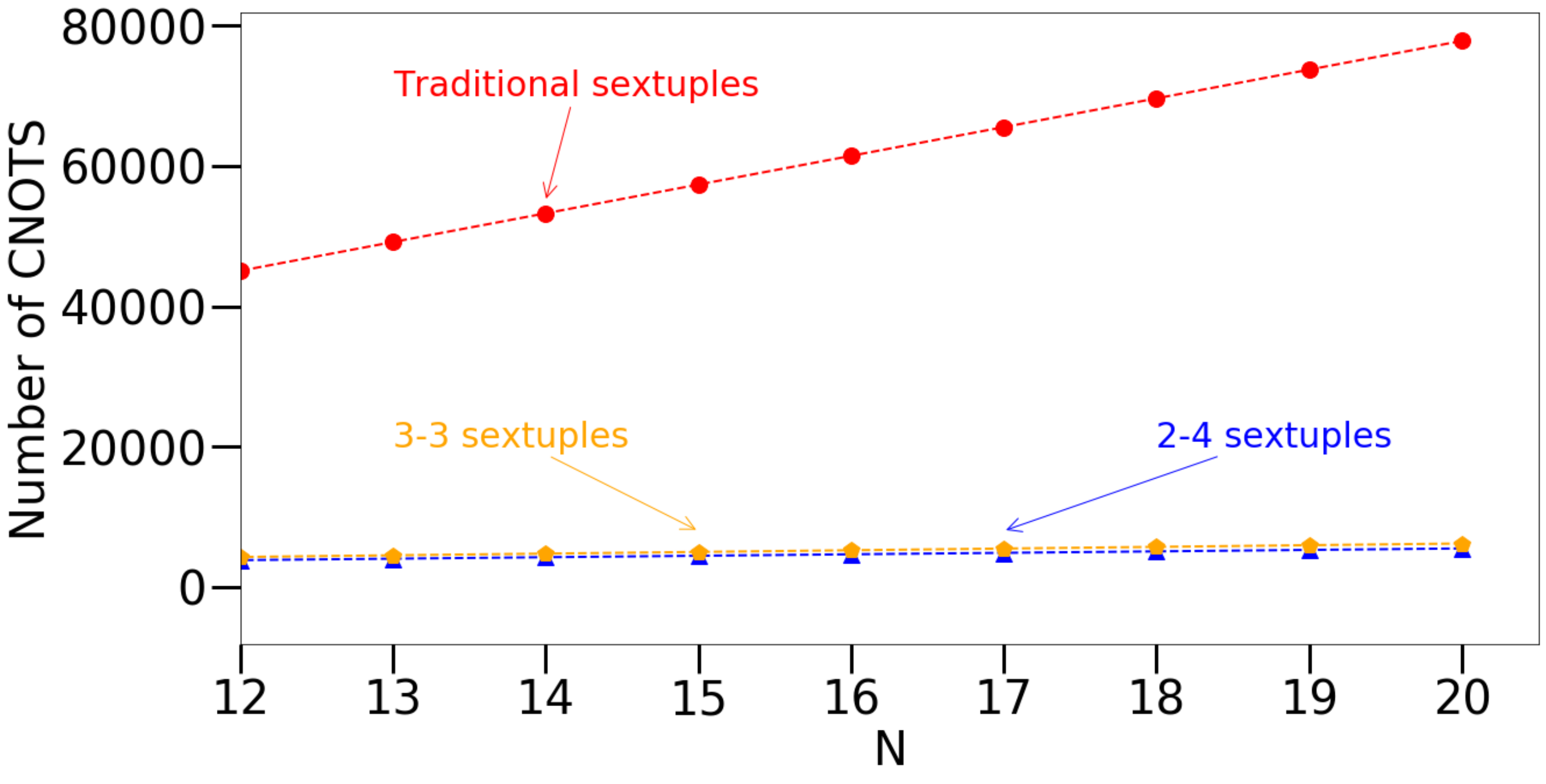}
        \caption{}
        \label{24sextuples}
        \end{subfigure}
    \hfill
    \begin{subfigure}{0.45\textwidth}
        \centering
        \includegraphics[width=\textwidth]{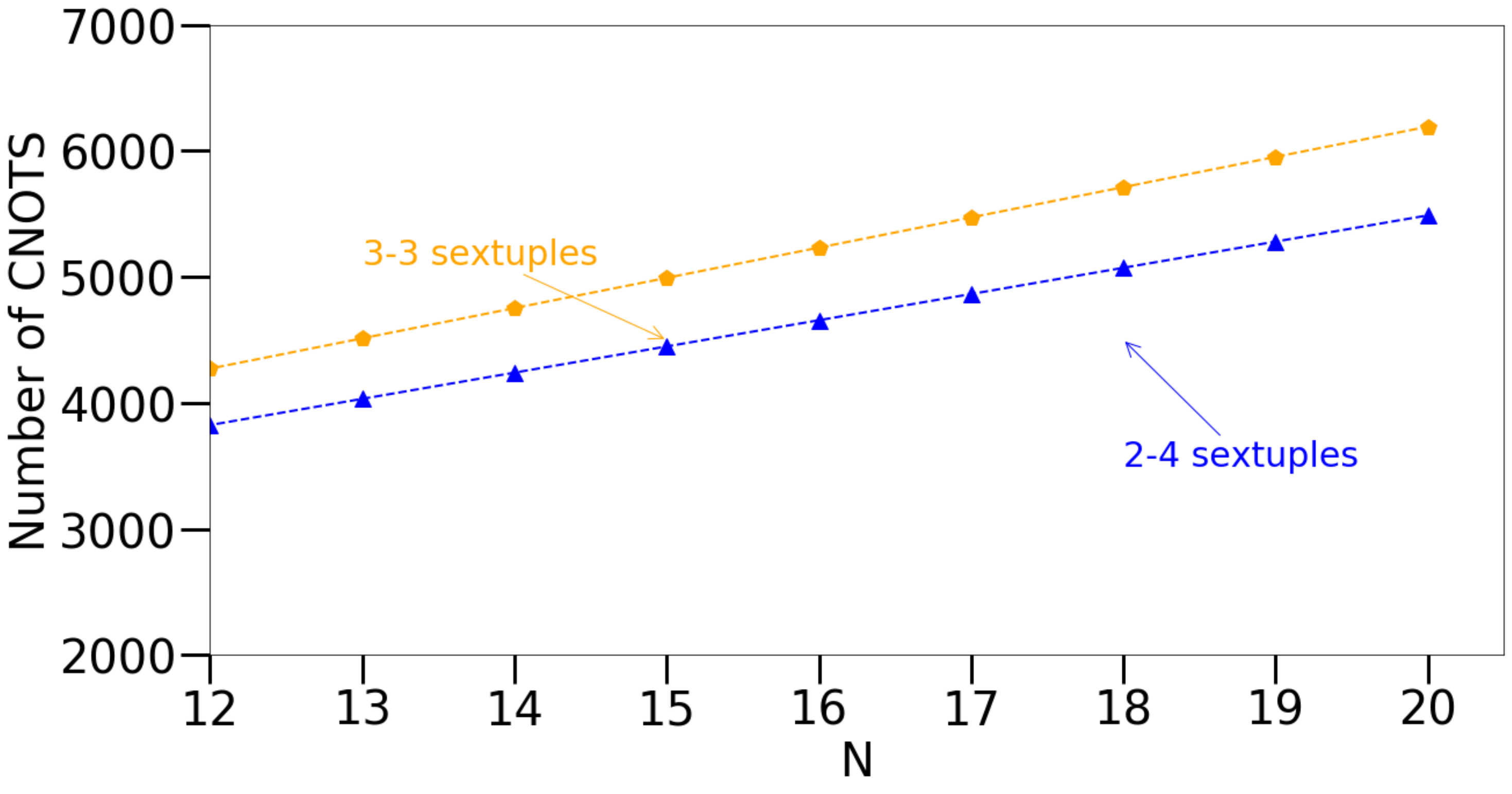}
        \caption{}
        \label{33sextuples}
        \end{subfigure}
    \caption{(\subref{24sextuples}) CNOT gate counts of traditional sextuples, the 3-3 sextuples, and the 2-4 sextuples. (\subref{33sextuples}) CNOT gate counts of the 3-3 sextuples and the 2-4 sextuples.}
\end{figure}

\section{Conclusion}
We have presented specific schemes to decompose high-rank UCC operators into low-rank singles and/or doubles, significantly reducing the number of CNOT gates needed to implement such circuits at the expense of using extra ancilla qubits. We have shown the proposed method is the most resource-friendly when the state preparation involves entangling a large number of qubits for a large system using high-rank UCC operators, such as quintuples and sextuples (or higher). It is anticipated such terms will be needed for strongly correlated molecules that are planned to be examined on quantum computers. 

For NISQ hardware, large numbers of two-qubit entangling gates are problematic. Generally, one wants to avoid having a large circuit depth due to noise, decoherence, and low fidelity. However, increasing the number of qubits in exchange for a circuit with less depth is favorable in the near term. For the construction of the specific scheme presented in this paper, we used the factorized form of the UCC ansatz, which was able to create the exact ground state wavevector using the method mentioned in \cite{xu_lee_freericks_2020}. Being able to decompose the UCC quadruples operator used in the state preparation for the ground state wavefunction for the 4-site Hubbard model at half filling, we manage to halve the total number of two-qubit gates. We anticipate that preparing strongly correlated states of larger systems, such as those studied in \cite{umrigar_2020}, will require higher-order UCC factors. Our approach should significantly reduce the gate counts for these circuits.
Similar strategies have been used to examine the decomposition of hardware efficient state-preparation protocols that preserve the particle number \cite{xanadu_2021}.

\section{Acknowledgements}
We acknowledge helpful discussions with Yan Wang, Ryan Bennink, and Eugene Dumitrescu. L. Xu and J. K. Freericks were supported by the U.S. Department of Energy, Office of Science, Office of Advanced Scientific Computing Research (ASCR), Quantum Computing Application Teams (QCATS) program, under field work proposal number ERKJ347. J. T. Lee was supported by the National Science Foundation under Grant No. DMR-1659532. J. K. Freericks was also supported by the McDevitt bequest at Georgetown University.

\bibliography{biblio}
\bibliographystyle{apsrev}

\hfill

\appendix

\section{Decomposition of the standard quadruple circuit}
In the extended figures, we show how one may implement the decomposed quad. Figure \ref{quad schematic} shows the generic order of UCC operators needed to implement the decomposed quad. It starts off with two doubles (indicated in figure \ref{quad schematic gate 1} and figure \ref{quad schematic gate 2}). Then, the doubly-controlled UCC doubles is applied. This circuit is given in figure \ref{doubly controlled double}. The following two blocks are the conjugate of the UCC factor of the previous two blocks. Thus, the circuit can be constructed by swapping $U_1 \leftrightarrow U_3$, and $U_2 \leftrightarrow U_4$.

\begin{figure}[htb]
\includegraphics[width=0.45\textwidth]{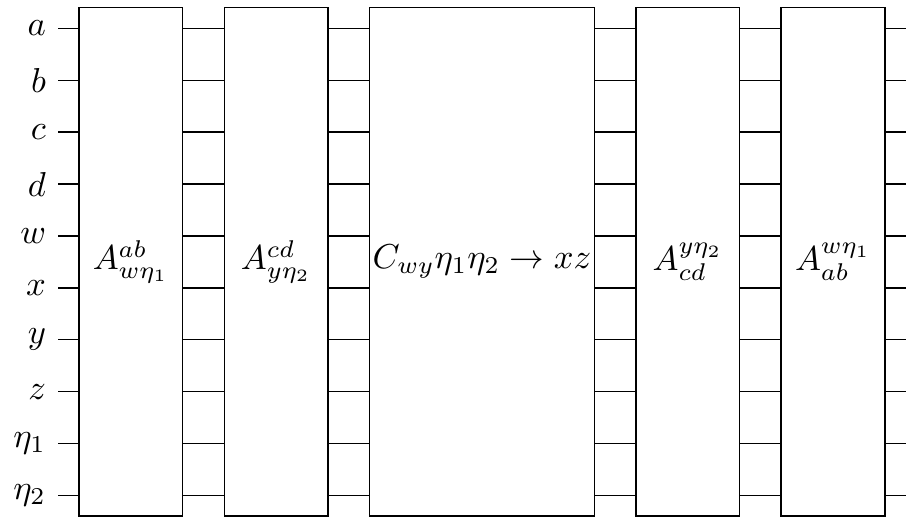}
\centering
\caption{Schematic of how the decomposed quad excitation would be implemented.}
\label{quad schematic}
\end{figure}

\begin{figure}[htb]
\includegraphics[width=0.45\textwidth]{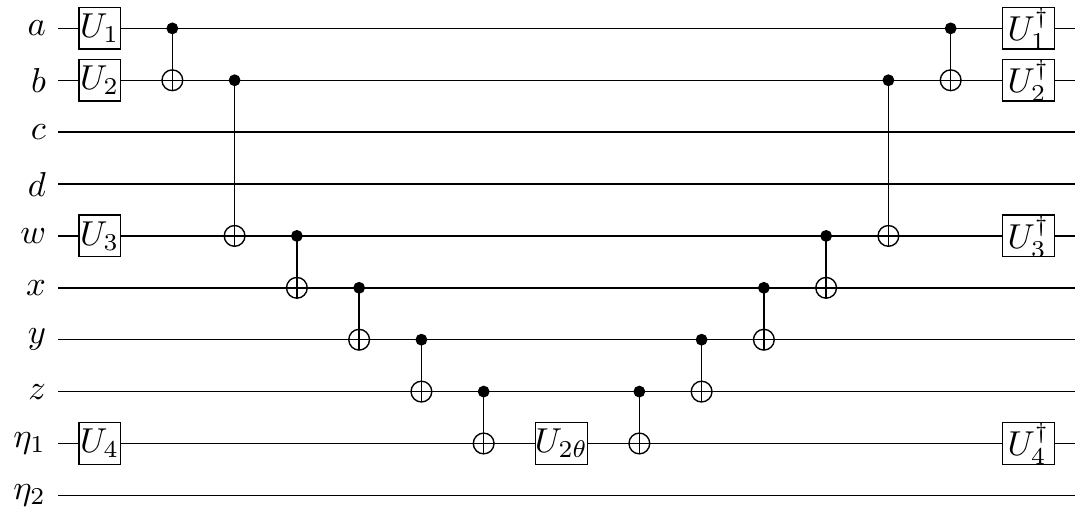}
\centering
\caption{Schematic of the first UCC factor in the decomposed quad.}
\label{quad schematic gate 1}
\end{figure}

\begin{figure}[htb]
\includegraphics[width=0.45\textwidth]{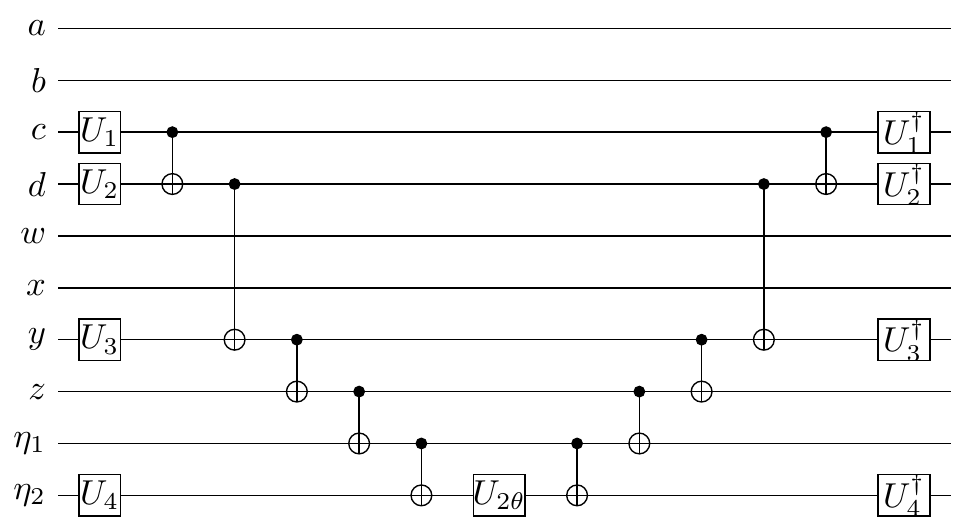}
\centering
\caption{Schematic of the second UCC factor in the decomposed quad.}
\label{quad schematic gate 2}
\end{figure}

\end{document}